\documentclass[onecolumn,a4paper,accepted=2024-12-20]{quantumarticle}
\pdfoutput=1
\usepackage{graphicx}
\usepackage[utf8]{inputenc} 
\usepackage[T1]{fontenc}
\usepackage[english]{babel}
\usepackage{xcolor}
\usepackage{amsmath}
\usepackage{amsthm}
\usepackage{amssymb}
\usepackage{fullpage}
\usepackage{multirow}
\usepackage{enumerate}
\usepackage{hyperref}
\usepackage{pdflscape}
\usepackage{array}
\usepackage{accents}

\newtheorem{Th}{Theorem}
\newtheorem{Prop}[Th]{Proposition}
\newtheorem{Cor}[Th]{Corollary}

\theoremstyle{remark}

\newtheorem{R}{Remark}

\theoremstyle{definition}
\newtheorem{Df}{Definition}
\newtheorem{Ex}{Example}

\def\tr{\textnormal{tr}}

\def\lin{\textnormal{lin}}
\def\conv{\textnormal{conv}}

\allowdisplaybreaks[1]

\begin{document}
\title{Can QBism exist without Q? Morphophoric measurements in generalised probabilistic theories}
\author{Anna Szymusiak}
\email{anna.szymusiak@uj.edu.pl}
\author{Wojciech S{\l}omczy\'{n}ski}
\email{wojciech.slomczynski@uj.edu.pl}
\address{Institute of Mathematics, Jagiellonian University, \L ojasiewicza 6, 30-348 Krak\'{o}w, Poland}

\begin{abstract}
In a Generalised Probabilistic Theory (GPT) equipped additionally with some extra geometric structure we define the morphophoric measurements as those for which the measurement map sending states to distributions of the measurement results is a similarity. In the quantum case, morphophoric measurements generalise the notion of a 2-design POVM, thus in particular that of a SIC-POVM. We show that the theory built on this class of measurements retains the chief features of the QBism approach to the basis of quantum mechanics. In particular, we demonstrate how to extend the primal equation (‘Urgleichung’) of QBism, designed for SIC-POVMs, to the morphophoric case of GPTs. In the latter setting, the equation takes a different, albeit more symmetric, form, but all the quantities that appear in it can be interpreted in probabilistic and operational terms, as in the original ‘Urgleichung’.
\end{abstract}

\maketitle

\section{Introduction}

In the last decade QBism
\cite{FucSch09,Appetal11,FucSch11,Ros11,FucSch13,Appetal17,DeBetal21,Fucetal22} has
become one of the most promising and original approaches to the foundations of
quantum mechanics. In \cite{SloSzy20} we showed that the core of this approach
remains largely untouched if SIC-POVMs, whose existence in an arbitrary
dimension, apparently necessary to develop a canonical version of QBism, has not yet been proven, are
replaced by the elements of the much larger class of morphophoric POVMs. In
the present paper we go much further, showing that \textsl{the basic ideas of QBism
are not limited to the quantum world, but rather are a fundamental
feature of a broad class of physical systems and their measurements}. Such
measurements must allow us not only to reconstruct the pre-measurement state
of the system from the probabilities of the measurement results (i.e. they
must be informationally complete), but also to reproduce the full geometry of
the state space in the same manner. That is why we called them
\textsl{morphophoric} (from Old Greek `form-bearing').

What are the minimum requirements to define morphophoricity? The
\textsl{measurement map} always sends the convex set of states to the probability
simplex of potential measurement results. Clearly, it follows from the classical \textsl{total probability
formula} that this map is \textsl{affine} \cite{Slo03,Wet21}. Namely, suppose that our
system is prepared in a state $x$ with probability $q$ and in a state $y$ with
probability $1-q$; thus it is in the mixed state $qx+(1-q)y$ before a given measurement. Now, when the measurement is performed on the system, producing the result $i$ with
probability $p_{i}(z)$, where $z$ is a pre-measurement state, then we get
$p_{i}(qx+(1-q)y)=\Pr(i)=\Pr(x)\Pr(i|x)+\Pr(y)\Pr(i|y)=qp_{i}(x)+(1-q)p_{i}
(y)$ from the total probability formula. The famous Mazur-Ulam theorem \cite{MazUla32} guarantees that each \textsl{similitude}
(sometimes also called \textsl{similarity}) is necessarily affine. Hence, assuming that the measurement
map is morphophoric (i.e. that it is a similitude), we in fact strengthen the
affinity assumption.

However, to define similitude one has to \textsl{impose some geometric
structure on the state space}, e.g. an inner product, norm, metric, angles or
orthogonality, see \cite{Baletal16} for the discussion of the relations between
these notions. All these structures are already naturally present in the
probability simplex. Introducing them into the state space, we can further
assume that they are preserved by the measurement map, and in consequence that
the state space and its image contained in the probability simplex are
\textsl{similar}. This allows us to represent faithfully the states as
probabilities, and the set of states as a subset of the probability simplex, which
seems to be the very essence of QBism, at least when it comes to its mathematical side. We call this set the \textsl{generalised qplex}.

In our opinion, the most natural platform to study in depth the concepts of
\textsl{morphophoricity}, \textsl{generalised QBism}, and their interrelationship is provided by the \textsl{operational} (or \textsl{statistical}) \textsl{approach
}to quanta, which dates back at least to the 1970s \cite{DavLew70,Lud70,Dav76}. The starting point for this approach was to look at any physical theory
(quantum mechanics in particular) as a probabilistic theory, in which we have
two structures connected by so-called statistical duality: on the one hand we
consider a convex set of (mixed) states, and on the other hand the dual set of
effects (simple, i.e. yes/no observables, or, in other words, yes/no questions), representing both the measurements that can be performed on the system, and the probabilities of the
measurement outcomes. The variant of this approach specialised to finite
dimensions, which makes it technically less challenging, but nevertheless still
sufficient in the area of quantum information, has been intensively studied in
the 21st century under the name of generalised probabilistic theories (GPT) 
introduced in \cite{Bar07}; for the discussion of the history of GPT, see
\cite[Sec. 1.1.1]{Lam17} and \cite[Sec. 1.1.5]{Wet21}, and of different variants of the name, see
\cite[p.4]{Pla21}. In the present paper, we follow several expository texts devoted
to GPT, including \cite{BarWil11,Udu12,JanLal13,Baretal14,BarWil16,MulMas16,Lam17,Wil18,Wil19,Pla21,Wet21,Tak22}.

To introduce a generalised probabilistic theory (see, e.g. \cite{Slo03,AliTou07,Lam17} for
the proofs of the facts below) let us start from a finite-dimensional real
vector space $V$ ordered by a convex cone $C$ (the set of \textsl{un-normalised
`states'}), which is: (a) \textsl{proper}, i.e. $C\cap-C=\{0\}$; (b)
\textsl{generating} (or \textsl{spanning}), i.e. $C-C=V$, and (c) \textsl{closed}. We write $x>0$ for every
$x\in C$ such that $x\neq0$. We also consider the dual space $V^{\ast}
:=\{g:V\rightarrow\mathbb{R}:g$ -- linear$\}$ ordered by the dual cone $C^{\ast
}:=\{g\in V^{\ast}:g(x)\geq0$ for every $x\in C\}$ of \textsl{positive
functionals} (the \textsl{effect cone}), which is also proper, generating and
closed. We call $g\in C^{\ast}$ \textsl{strictly positive} when $g(x)>0$ for all $x>0$. It is well known that $\{g\in C^{\ast}:g$ is
strictly positive$\}=\operatorname*{int} C^{\ast}$.

Moreover, we assume that $C$ has a base $B$, which is interpreted as the
\textsl{set of states} of the system. Note that $B$ is necessarily compact and $B=\{x\in C:e(x)=1\}$ for a unique $e \in \operatorname* {int} C^{\ast}$. We call $e$
the \textsl{unit effect }(also called the \textsl{charge functional} or 
\textsl{strength functional} in the literature) and $\left( V,C,e\right)  $ 
or $\left( V,C,B\right)$ an \textsl{abstract state space}. The extremal elements 
of $B$ are interpreted as 
\textsl{pure states}. The elements of the interval $[0,e]=\{g\in V^{\ast}:0\leq
g\leq e\}$ are called \textsl{effects}. In this paper we work under the assumption that all `mathematical'
effects are physically possible, known as the \textsl{no restriction hypothesis}
\cite{Udu12}. When $x$ is a state and $g$ is an effect, then $g(x)\in
\lbrack0,1]$ can be interpreted as the probability that the answer to the
question $g$ is `yes', assuming that the system is in the state $x$. Observe that $(V,C)$ and its dual $(V^{\ast},C^{\ast})$ are normed order vector spaces with
the norms given by $\left\|  x\right\|  _{V}:=\inf\{e(w)-e(z):x=w-z,w,z>0\}$
for $x\in V$, and $\left\|  g\right\|  _{V^{\ast}}:=\min\{\lambda>0:-\lambda
e\leq g\leq\lambda e\}$ for $g\in V^{\ast}$, respectively. The former is a
\textsl{base norm space}, the latter a \textsl{unit norm space}, forming together so-called \textsl{statistical duality}
\cite{SinStu92,Buscheta16}. They are
related as follows: $x\in C$ if and only if $g(x)\geq0$ for all $g \in C^{\ast}$, $\left\|
x\right\|  _{V}=\max\{\left|  g(x)\right|  :-e\leq g\leq e, g \in V^{\ast}\}$ for $x\in V$,
and $\left\|  g\right\|  _{V^{\ast}}=\max\{\left|  g(x)\right|  :x\in B\}$ for
$g\in V^{\ast}$.

Further, we define a \textsl{measurement} as a sequence of nonzero effects
$(\pi_{j})_{j=1}^{n}$ such that ${\textstyle\sum\nolimits_{j=1}^{n}}
\pi_{j}=e$. Then $\pi_{j}(x)$ is interpreted as the probability that the
measurement outcome is $j$ if the pre-measurement state is $x\in B$. Clearly, the \textsl{measurement map} $\pi:B\ni x\longmapsto(\pi_{j}(x))_{j=1}^{n}\in\Delta_{n}$
is an affine map from $B$ into the \textsl{probability simplex} $\Delta
_{n}:=\{p\in\mathbb{R}^{n}:p_{j}\geq0$ for $j=1,\ldots,n$ and $
{\textstyle\sum\nolimits_{j=1}^{n}}
p_{j}=1\}$.

This is a framework for a \textsl{generalised probabilistic theory} (\textsl{GPT}).
Clearly, several variants of this approach are possible. For instance, one can
start from more elementary objects, e.g. from abstract convex sets like convex
structures (modules) representing the states of the system, but the celebrated
Stone embedding theorem \cite{Sto49} guarantees that under some mild additional assumptions they can
always be embedded into an abstract state space.

Our first step is to introduce \textsl{geometry} into $B$, necessarily Euclidean (this follows from morphophori\-city since the geometry of the
probability simplex is such), e.g. as an inner product $\left\langle
\cdot,\cdot\right\rangle _{0}$ in the vector subspace $V_{0}:=e^{-1}
(0)=V_{1}-V_{1}$, where $V_{1}:=e^{-1}(1)$ is the affine space generated by
$B$. Now, to endow the full space $V$ with an inner product extending
appropriately $\left\langle \cdot,\cdot\right\rangle _{0}$, we need two
further ingredients: a \textsl{distinguished `state'} $m\in V_{1}$ and a
\textsl{size parameter} $\mu>0$. Then, the inner product $\left\langle
\cdot,\cdot\right\rangle _{m,\mu}$ in $V$ can be uniquely defined by the
following two conditions: $m\perp V_{0}$ and $\left\langle m,m\right\rangle
_{m,\mu}=\mu$. (Note that $V$ is then the orthogonal sum of $V_{0}$ and $\mathbb{R}m$.) Namely, it is easy to show
that $\left\langle \cdot,\cdot\right\rangle _{m,\mu}$ is given by 
\begin{equation}
\label{innpro}
\left\langle x,y\right\rangle _{m,\mu}:=\left\langle
x-e(x)m,y-e(y)m\right\rangle _{0}+e(x)e(y)\mu
\end{equation}
for $x,y \in V$.  

On the other hand, any inner product on $V$ being an extension of $\langle
\cdot,\cdot\rangle _{0}$ must have the form $\langle
\cdot,\cdot\rangle _{m,\mu}$, where $m$ is the unique element of $V_{1}$ orthogonal to $V_{0}$ and $\mu$ is the square of its norm.
Summarising this reasoning, there are two alternative and equivalent ways to introduce 
Euclidean geometry into $V$. We can choose either an inner product in $V$, or a triple
consisting of an inner product in $V_0$, an element of $V_1$, and a positive number.
We call this structure a \textsl{geometric generalised probabilistic theory} (\textsl{GGPT}).
In the sequel, we shall always assume that a GGPT is given.

What is crucial here it is the fact that, as we shall see, just two ingredients:
\begin{itemize}
\item  a state space,
\item  a Euclidean geometry on states,
\end{itemize}
are enough to introduce morphophoricity and, \textsl{ipso facto}, to prepare a recipe for QBism-like structures in a
general setting. On the other hand, some characteristics of these structures depend on two further ingredients:
\begin{itemize}
\item  a distinguished state,
\item  a size parameter,
\end{itemize}
or, equivalently, on the geometry of the (full) state space. 
\smallskip

The above components are necessary. However, to make these structures more palatable, 
we need something else. First of all, we would like to connect \textsl{order} and \textsl{geometry} of the state space $B$ in an
appropriate manner. Namely, we consider two forms of compatibility of these structures: \textsl{infra-} and \textsl{supra-duality} of the positive cone. In the GGPT case they can be characterised by additional conditions that can be imposed on the parameters of the theory, $m$ and $\mu$, see Secs. \ref{Compatibility with order} and \ref{Compatibility with order for GGPTs}. In particular, Theorem \ref{compat} gives necessary and sufficient conditions for the size parameter $\mu$ that make the cone infra- or supra-dual. Naturally, it provides also a condition for the state space to be \textsl{self-dual}, i.e. both infra- \textsl{and} supra-dual. Observe that one can always make the space infra- or supra-dual (but not always both!), by changing appropriately the parameter $\mu$, so the assumptions of infra- or supra-duality are not very restrictive.
However, in the self-dual case, the size parameter $\mu$ is predetermined, and the resulting theory is more symmetric and elegant. 

In Sec. \ref{Geometrical properties of set of states} we continue the analysis of the geometric properties of the set of states. As multiplying the scalar product by some positive constant does not change its geometry, we introduce a `dimensionless' parameter, the \textsl{space constant}, $\chi / \mu$, where $\chi$ is the maximal norm of the Bloch vector of the state. This quantity gives us some information on the geometry but not on the size of the set of states, and can be described in several different ways as the function of the maximal angle between the Bloch vectors of pure states (Proposition \ref{Funine}), the orthogonal dimension of the state space (Proposition \ref{ort}), and the maximal entropy of a state (Propositions \ref{decent} \& \ref{speent}). This also explains why we can interpret $m$ as the \textsl{maximally mixed state}.

As some of our results require the self-duality of the state space, we analyse this assumption more carefully in Sec.~\ref{Self-duality}, proving that self-duality depends only on the order structure of the space (i.e. the cone $C$) but not on the particular choice of a base $B$ or a unit effect $e$ (Theorem~\ref{Self-dual equivalent}). Moreover, we show that it is possible to express the (functional analytic) self-duality of the state space in the language of the (geometric) self-duality of the set of states $B$ with respect to the sphere with center $m$ and radius $\sqrt{\mu}$ (Theorem \ref{selfdual}). Two extreme examples of self-dual GGPTs fulfilling the assumptions of so-called spectrality are the theories where $B$ is the $N$-dimensional \textsl{ball} or the $N$-dimensional \textsl{regular simplex}, representing the \textsl{classical} probability theory (Proposition \ref{extreme}). We also consider in Sec.~\ref{exGGPT} two other examples: \textsl{quantum} GGPT, which is self-dual, and the one where $B$ is a \textsl{regular polygon}, which is either self-dual for an odd number of vertices, or not self-dual when the number is even. This includes the well-known \textsl{gbit} or \textsl{Boxworld} GGPT where $B$ is the \textsl{square}. At this point the lists are ready, so we can introduce our knights: the morphophoric measurements.

The notion of morphophoricity, which is, as has been said, the core of this paper, is discussed in Sec.~\ref{Morphophoricity}. The key result concerning morphophoric measurements is their characterisation in the language of \textsl{tight frames} (Theorem \ref{morphTF}). Namely, the sufficient and necessary condition for a measurement to be morphophoric is that (on par with several other equivalent conditions) the orthogonal projections of the measurement effects onto the vector subspace of the functionals equal to  zero at $m$ constitute a tight (and balanced) frame. The proof of this statement is one of the crucial points of the whole paper. This fact implies also that morphophoric measurements exist for all GGPTs (Theorem \ref{existence}). The next results (Proposition \ref{str1}, Theorems \ref{str2} and \ref{str3}) give some insight into the structure of this family of measurements.

The morphophoric measurement plays a double role in the \textsl{geometry of the generalised qplex}, which we study in detail in Sec.~\ref{General case}. Firstly, the effects of this measurement generate, via the isomorphism between the dual space and the original state space, a collection of un-normalised vectors $\{v_j\}_{j=1,\dots,n}$ that lie, for a supra-dual space, in the positive cone $C$. Their orthogonal projections on the subspace $V_0$ constitute, as it was said above, a tight frame. 
Its image by the measurement constitutes a tight frame as well, this time for the linear subspace corresponding to the \textsl{primal affine space}, i.e. the affine span of the image of $B$ by the measurement.
This fact is also the sufficient and necessary condition for morphophoricity (Theorem \ref{morfra}). 
Note that this frame is the homothetic image of the orthonormal projection of the canonical orthogonal basis (i.e. the vertices of the probability simplex) of $\mathbb{R}^n$ onto the linear subspace mentioned above (Proposition \ref{proj}).
Secondly, the vectors $v_j$ normalised appropriately and transformed by the same measurement map give the so-called \textsl{basis distributions}. Their convex hull defines the \textsl{basis polytope} $D$ contained in the generalised qplex. On the other hand, the qplex is a subset of the \textsl{primal polytope} $\Delta$, being the intersection of the primal affine space and the probability simplex. Both polytopes are dual (Theorem \ref{duapol}) with respect to the sphere with centre at $c$ and of radius $\sqrt{\alpha\mu}$, where $c \in D$ is the \textsl{central distribution}, i.e. the image of the distinguished state $m$ by the measurement map, $\sqrt{\alpha}$ is the \textsl{similarity ratio} of the measurement map, and $\alpha\mu$, the \textsl{measurement constant}, is another `dimensionless' parameter characterising the measurement.

This geometric picture is even clearer in the case of \textsl{regular morphophoric measurements in self-dual spaces} (Definition \ref{regular}) we analyse in Sec. \ref{Regular measurements}. These measurements are the GGPT analogues of the rank-1 equal norm POVMs generated by 2-designs in quantum case. Note that the SIC-POVMs used in QBism also belong to this class. Firstly, in this case all the constants of the theory are related by a surprisingly simple formula (Theorem \ref{formula}). Namely,
\begin{equation*}
\text{measurement constant}=\frac{\text{space constant}}{\text{measurement
dimension}\times\text{space dimension}}\, , 
\end{equation*}
where the \textsl{measurement dimension} $n$ is necessarily larger than the \textsl{space dimension} (i.e. $\dim B = \dim{V}_{0}$). Secondly, the values of several quantities such as the bounds in the \textsl{fundamental inequalities} for probabilities in the qplex (Proposition \ref{Funinepro}) and the \textsl{radii} of the \textsl{inner ball} inscribed in the primal polytope or the \textsl{outer ball} circumscribed about the basis polytope (Proposition \ref{spheres and polytopes}) gain a new and rather unexpected interpretation in the light of the general theory, revealing in a sense their hidden meaning.

We presented various examples of quantum morphophoric measurements in \cite{SloSzy20}, including e.g. the POVMs based on MUB-like 2-designs. In Sec. \ref{Morphophoric measurements - examples} we give several examples of morphophoric measurements in more `exotic' GGPTs: the Boxworld (gbit), pentagonal, and ball GGPT.

However, as Boge wrote recently ``today, QBists’ main focus is on what they call the Urgleichung'' \cite{Bog22}. Accordingly, Sec. \ref{The primal equation} devoted to the \textsl{primal equation} (or, in other words, \textsl{Urgleichung}) is also crucial for us. In \cite{SloSzy20} we show that the morphophoricity of a 2-design POVM representing the measurement \textsl{`in the sky'} is equivalent, under the additional assumption that the state of the system after this counterfactual measurement is described by the generalised L\"{u}ders instrument $\Lambda$, to a form of the primal equation only slightly changed from the original one \cite[eqs. (21) \& (22)]{SloSzy20}. We showed that this `new' Urgleichung can be presented in the quantum case in a purely probabilistic way, also in the general situation, i.e. not necessarily for rank-1 POVMs consisting of effects of equal trace \cite[Theorem 18]{SloSzy20}. 

In the current paper we prove that the same situation holds for a \textsl{general GGPT} under fairly unrestrictive conditions:

\begin{enumerate}[I.]
\item the GGPT $(V,C,e, m,\mu, {\langle \cdot,\cdot \rangle}_0)$ is \textsl{supra-dual};
\item the instrument $\Lambda$ describing the state of the system after the measurement $\pi$ is \textsl{balanced at} $m$.
\end{enumerate}

Now the measurement $\xi$ `on the ground' is \textsl{arbitrary}, whereas morphophoricity of the counterfactual `in the sky' measurement $\pi$ is then equivalent to the fact that the primal equation holds. Note that it is presented in  a concise form as Theorem \ref{pe} and in a purely probabilistic version as Corollary \ref{pepr}. Now, let us briefly discuss these two assumptions above. As we have already mentioned, the former assumption is not restrictive at all since we are able to make the GGPT supra-dual by taking the parameter $\mu$ small enough, see Remark \ref{makesupra}. On the other hand, we discuss the latter assumption in Sec. \ref{Instruments} in detail, showing that the fact that the instrument $\Lambda$ is balanced at $m$ is strictly related to the problem of \textsl{retrodiction} and the \textsl{`Bayesian behaviour'} of $\Lambda$, but only at one particular point: the `equilibrium' $m$.

So, rather unexpectedly, the primal equation turns out to be related not only to the \textsl{total probability formula}, as we shall observe later in a particular, though important, case, but also directly to the classical \textsl{Bayes formula}. So the letter `B' in the word `QBism' finds (again) its strong justification. But is it really so with the letter `Q'? It turns out that the answer to this question can be both `yes' and `no'. On the one hand, the SIC-POVMs and the generalised L\"{u}ders instrument used in QBism theory undoubtedly have certain specific and distinguishing features. On the other hand, as we shall see in the present paper, almost the entire mathematical part of this theory, including the geometry of the generalised qplex and the primal equation, is not only preserved in the general setting, but  the values of its specific parameters also get a deeper and clearer explanation in a broader context. In the light of this, (generalised) QBism appears to be not only an alternative approach to the foundations of quantum mechanics, but also a \textsl{general theory of certain measurements with specific properties}.

But how is it possible that this fact has gone unnoticed for so many years? We think it was partly because both the SIC-POVMs and the generalised  L\"{u}ders instrument used in QBism have some special properties, which cause the generalised primal equation to take a form somewhat different from \eqref{probab} or \eqref{conditional}, resembling a modified total probability formula with surprising and somewhat `magical' parameters\footnote{However, as Peter Parker puts it, `You know what's cooler than magic?... Math!' (\textsl{Spider-Man: No Way Home}). In fact, the parameters that appear in \eqref{gur}, $A$ and $G$, are not `magical' at all, but depend on the measurement constant $\alpha\mu$ and the measurement dimension $n$. On the other hand, \eqref{probab} and \eqref{conditional} contain no parameters but only the dimension of the set of states, which is equal to the dimension of the primal affine space, $\dim \mathcal A$.}. Namely, a SIC-POVM is \textsl{unbiased}, i.e. all the probabilities at $m$ are equal, and the generalised L\"{u}ders instrument for rank-1 POVMs is \textsl{canonical} (Definition \ref{canon}), i.e. not only is it balanced at $m$ but, what is equally important here, \textsl{the posterior states for this instrument are independent of priors} (Proposition \ref{IE}). In consequence, the conditional probabilities that appear in the primal equation are also independent of priors. This creates the illusion that all probabilities in \eqref{gur} are always calculated at the same (and arbitrary) point in the state space. However, this is only the case in this particular situation, where we can easily deduce from our generalised primal equation a \textsl{total probability like} formula (Proposition \ref{pec}), called by Fuchs and Schack \cite{FucSch09,FucSch11} the \textsl{Generalised Urgleichung}. Hence, we can obtain the classical total probability formula (for the classical measurement) or the quantum Urgleichung (for SIC-POVMs) as special cases, along with many related formulas for other types of GGPTs, like the  \textsl{Quaternionic Urgleichung}, see \cite[Sec. 5.2 \& 5.4]{Gra11}. However, the  situation is more complicated even for non rank-1 morphophoric POVMs and the generalised L\"{u}ders  instrument, see \cite{Lud06}, as in this case the posterior states are no longer independent of priors. Then, the only acceptable general form of the primal equation seems to be that described by Theorem \ref{pe} or Corollary \ref{pepr}.

\section{Preliminaries}

\subsection{Compatibility with order}
\label{Compatibility with order} 

The question arises whether or not the inner products on an abstract state space $(V,C,e)$
discussed in the previous section and its order structure are compatible. For
ordered real vector spaces we distinguish between three types of such compatibility. (Here, and below, we follow the terminology of Iusem \& Seeger \cite{IusSee08}.)

\begin{Df}
Let $V$ be a finite-dimensional real vector space ordered by a proper spanning
closed convex cone $C$ and endowed with an inner product $\left\langle \cdot
,\cdot\right\rangle :V\times V\rightarrow\mathbb{R}$. Then, with respect to $\left\langle \cdot,\cdot\right\rangle$, the cone $C$ is called 

\begin{itemize}
\item \textsl{infra-dual} if and only if $\left\langle x,y\right\rangle \geq0$
for $x,y\in C$;

\item \textsl{supra-dual} if and only if $\left\langle
x,y\right\rangle \geq0$ for all $x\in C$ implies $y\in C$;

\item \textsl{self-dual} if and only if $\left\langle
x,y\right\rangle \geq0$ for all $x\in C$ is equivalent to $y\in C$.
\end{itemize}
\end{Df}

Clearly, $C$ is self-dual if and only if
it is infra-dual and supra-dual. Define the \textsl{positive dual} cone of $C$ in $V$ as
$C^{+}:=\{y\in V:\left\langle x,y\right\rangle \geq0$ for all $x\in C\}$. (Clearly, $C^{++}=C$.) Then
$C$ is infra-dual if and only if
$C\subset C^{+}$, supra-dual if and only if $C^{+}\subset C$, and self-dual if and only if $C^{+}=C$. We can also express these properties in the language of the dual space $V^*$. Let $T:V\rightarrow V^{\ast}$ be the \textsl{isometric
linear isomorphism related with }$\left\langle \cdot,\cdot\right\rangle$ by
$T(y)(x):=\left\langle x,y\right\rangle$ for $x,y\in V$. Then
$T(C^{+})=C^{\ast}$. Thus, $C$ is
infra-dual if and only if $T(C)\subset C^{\ast}$, supra-dual if and only if
$C^{\ast}\subset T(C)$, and self-dual if and only if $C^{\ast}=T(C)$. 
In the last case $T$ is also an order isomorphism.

\begin{figure}[htb]
	\includegraphics[scale=0.18]{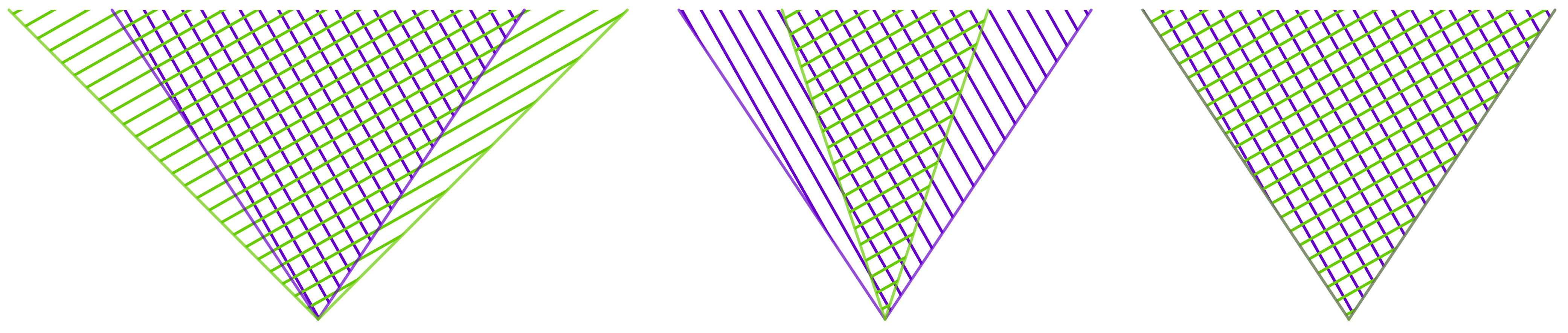}
	\centering
	\caption{The cone $C$ (purple) and the positive dual cone $C^+$ (green) in the infra-dual, supra-dual and self-dual case, respectively.}
\end{figure}

Consequently, we
call $(V,C)$ or a GPT  \textsl{self-dual} if there exists $\left\langle \cdot
,\cdot\right\rangle :V\times V\rightarrow\mathbb{R}$ making $C$  self-dual
with respect to this product. Equivalently, $(V,C)$ is self-dual if and only if there exists an order isomorphism  $T:V\rightarrow V^{\ast}$ such that 
$T(x)(y)=T(y)(x)$ and $T(x)(x) \geq 0$ for $x,y \in V$.

\begin{R}
Some authors use the term \textsl{strongly self-dual} for self-duality in our sense. Then by \textsl{weak self-duality} they mean the mere  existence of an order linear isomorphism  $T:V\rightarrow V^{\ast}$. As we consider only the first type of self-duality here, we call it simply self-duality.
\end{R}

\subsection{Compatibility with order for GGPTs}
\label{Compatibility with order for GGPTs}
Let us consider the GGPT generated by 
an abstract state space $\left(  V,C,e\right)$ with an inner product
$\left\langle \cdot,\cdot\right\rangle _{0}:V_{0}\times V_{0}\rightarrow
\mathbb{R}$, a distinguished point $m$ such that $e(m)=1$, and a size
parameter $\mu>0$. Clearly, in this case $e$ is given by $e(x_{0}+\lambda m)=\lambda$ for
$x_{0}\in V_{0}$, $\lambda\in\mathbb{R}$, and $V_{1}=V_{0}+m$. An affine map $V_{1}\ni x\rightarrow x_{m}:=x-m\in V_{0}$
is called a \textsl{Bloch representation}. Let the inner product $\langle
\cdot,\cdot \rangle _{m,\mu}$ in $V$ be given by \eqref{innpro}. Recall that $\left\langle
m,x\right\rangle _{m,\mu}=0$ for $x\in V_{0}$ and $\left\langle
m,m\right\rangle _{m,\mu}=\mu$.
The following statements are elementary.

\begin{Prop}
\label{ProInn}
Let $x,y\in V$. Then

\begin{enumerate}[i.]
\item\label{minn} $\langle x,m\rangle _{m,\mu}=\mu e(x)$,

\item\label{inn0} $\langle x,y\rangle _{m,\mu}=\langle x_{m}
,y_{m}\rangle _{0}+\mu$ for $x,y\in V_{1}$,

\item $\left\|  x\right\|  _{m,\mu}^{2}=\left\|  x_{m}\right\|  _{0}^{2}+\mu$
for $x\in V_{1}$,

\item\label{norms} $\left\|  x\right\|  _{m,\mu}^{2}\geq\mu$ and $\left\|  x\right\|
_{m,\mu}^{2}=\mu$ if and only if $x=m$ for $x\in V_{1}$.
\end{enumerate}
\end{Prop}

From Proposition \ref{ProInn}.\ref{minn}. we can see immediately that the distinguished
state $m$ is just a scaled incarnation of the distinguished unit effect $e$ into the realm of state space. 

\begin{Prop}
\label{mmu} Let $T_{m,\mu}:V\rightarrow V^{\ast}$ be the isometric isomorphism related to
the inner product $\langle \cdot,\cdot\rangle _{m,\mu}$, i.e.
$T_{m,\mu}(x)(y):= \langle x,y\rangle _{m,\mu}$
for $x,y\in V$.  Then 
$m=\mu T_{m,\mu}^{-1}(e)$ and $\mu=1/e(T_{m,\mu}^{-1}(e))$.
\end{Prop}

From now on, we assume that the distinguished point $m$ lies in the relative interior of the set of states $B$ with respect to $V_1$. We shall see in Remark \ref{supra} that this fact actually follows from the assumptions of the supra-duality of $C$, widely used in this paper, though some of our results are also true for an arbitrary `state' $m$ lying in $V_1$, see e.g. Theorem \ref{morphTF}. 

Below, we  establish the equivalent conditions for $C$ to be, respectively, infra-dual, supra-dual, and self-dual with respect to the inner product $\left\langle \cdot,\cdot\right\rangle_{m,\mu}$. For every choice of the geometry of states given by the inner product $\left\langle \cdot,\cdot\right\rangle _{0}$ and the distinguished state $m$ as above, taking appropriately the size parameter $\mu$ one can always make the cone  infra- or supra-dual. However, you can't have your cake and eat it too, at least not in every case. Such choice of the size parameter is possible (and unique!) only for special kind of GPTs, namely, where the state space is strongly self-dual.

\begin{Th}
\label{compat}
Let $m\in\operatorname*{int}_{V_{1}}
B$, $\mu> 0$. Then $C$ with respect to $\left\langle \cdot,\cdot\right\rangle _{m,\mu}$ is:

\begin{enumerate}[i.]
\item\label{id} infra-dual if and only if $\mu_i:=-\min_{x,y\in\operatorname*{ex}B}\left\langle
x_{m},y_{m}\right\rangle _{0}\leq\mu$;

\item\label{suprad} supra-dual if and only if $\mu\leq-\max_{x\in\partial B}\min
_{y\in\operatorname*{ex}B}\left\langle x_{m},y_{m}\right\rangle _{0}=:\mu_s$;

\item\label{sd} self-dual if and only if $\min_{y\in\operatorname*{ex}B}\left\langle
x_{m},y_{m}\right\rangle _{0}=\operatorname*{const} (x\in\partial B)=-\mu$,
\end{enumerate}
where $\partial B$ is the relative boundary of $B$ in $V_1$.
\end{Th}

The proof is presented in the Appendix \ref{AppendixProof}.
In the above situation we also say that a given GGPT is, respectively, \textsl{infra-dual}, \textsl{supra-dual} or \textsl{self-dual}, or that the given form of duality of $(V,C)$ is realised by this GGPT.

\begin{R}\label{makesupra}
Note that $\mu_s>0$. Indeed, if $\max_{x\in\partial
B}\min_{y\in\operatorname*{ex}B}\left\langle x_{m},y_{m}\right\rangle _{0}
\geq0$, then there would exist $x\in\partial B$ such that for all $y\in B$ we
have $\left\langle x_{m},y_{m}\right\rangle _{0}\geq0$. In consequence,
$m\notin\operatorname*{int}_{V_{1}}B$, a contradiction. Thus, taking
$\mu$ small enough, we can always make $C$ supra-dual with respect to $\langle \cdot,\cdot\rangle _{m,\mu}$. On the other hand, we can always make it infra-dual, taking $\mu$ large enough. However, intervals of infra- and supra-duality either intersect at one point, and then one can make a GGPT self-dual by setting the parameter $\mu$ accordingly, or they are disjoint and it is not possible for any choice of $\mu$, see e.g. Ex. (C) from Sec. \ref{exGGPT}.
Which of these situations occurs depends on the geometry of $B$, i.e. on $\left\langle \cdot,\cdot\right\rangle _0$ and $m$.
In particular, if $C$ is self-dual with respect to $\left\langle \cdot,\cdot\right\rangle _{m,\mu}$, then  
$\mu=-\min_{x,y\in\operatorname*{ex}B}\left\langle
x_{m},y_{m}\right\rangle _{0}$.
\end{R}
\begin{figure}[htb]
	\includegraphics[scale=0.22]{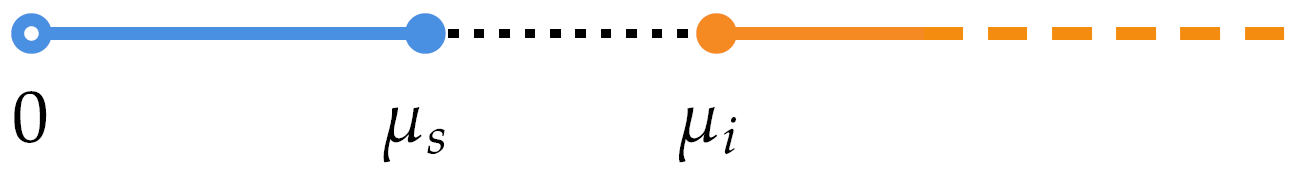}
	\centering
	\caption{The ranges of parameter $\mu$ making a GGPT supra-dual (blue) or infra-dual (orange). If there is no gap between $\mu_s$ and $\mu_i$ the GGPT can be made self-dual.}
\end{figure}
\begin{R}
\label{supra}
If $C$ is supra-dual with respect to $\langle \cdot,\cdot\rangle _{m,\mu}$, then $T^{-1}_{m,\mu}(C^\ast) \subset C$. As $e \in \operatorname*{int} C^\ast$ and $T_{m,\mu}:V\rightarrow V^{\ast}$ is continuous, we have $T_{m,\mu}^{-1}(e) \in T^{-1}_{m,\mu}(\operatorname*{int}C^\ast)\subset\operatorname*{int}T_{m,\mu}^{-1}(C^\ast)\subset \operatorname*{int}C$. From Proposition \ref{mmu} we deduce that $m \in \operatorname*{int}C$. Thus, $m\in\operatorname*{int}_{V_{1}}B$. 
In consequence, applying Theorem \ref{compat}.\ref{suprad}., we deduce that $C$ is supra-dual with respect to $\langle \cdot,\cdot\rangle _{m,\mu}$ if and only if $m\in\operatorname*{int}_{V_{1}}B$ and $\mu\leq\mu_s$. Moreover, if $C$ is self-dual with respect to $\langle \cdot,\cdot\rangle _{m,\mu}$, then $\mu$ is uniquely determined by $\langle \cdot,\cdot\rangle _{0}$ and $m$.
\end{R}

\subsection{Geometrical properties of set of states}
\label{Geometrical properties of set of states}
There are two sources of geometrical structure of generalised qplexes. The first is the original geometry of a GGPT, the second is morphophoricity of the measurement map that transforms this geometry into the geometry of subset of the probability simplex. In this section we analyse the former one.

Let a GGPT be fixed. Now, we introduce another positive constant
$\chi:=\max\{  \left\|  x_{m}\right\|  _{0}^{2}:x\in B\}$, which plays a crucial role in analysing its geometry. It depends only on $\langle \cdot,\cdot\rangle _{0}$ and $m\in\operatorname*{int} _{V_{1}}B$, but not on $\mu$.

Moreover, the states of maximal norm are
necessarily pure.

\begin{Prop}
\label{extreme points}
$M(B):=M_{m}(B):=\{  x \in B :\left\|  x_m\right\|  _{0}^{2}=\chi\}  \subset \operatorname*{ex}B$.
\end{Prop}

\begin{proof}
It follows from the definition that elements of $M(B)$ are \textsl{bare
points} of $B$, i.e. there is a closed ball in $V_1$ containing $B$, namely $B(m,\chi)$, with  $M(B)$ contained in its surface. Hence they are extremal points of $B$ \cite{Ber67}.
\end{proof}

For many GGPTs the equality $M(B)=\operatorname*{ex}B$ is true. We call such GGPT \textsl{equinorm}. This property depends only on $\langle \cdot,\cdot\rangle _{0}$ and $m$ but not on $\mu$.

\begin{R}
 In particular, a GGPT is equinorm, if it is \textsl{symmetric}, i.e, if for every pair $x,y \in \operatorname*{ex}B$ there exists an isometry $L:V \rightarrow V$ such that $L(x)=y$, or, equivalently, an isometry $L_0:V_0 \rightarrow V_0$ such that $L_0(x_m)=y_m$.
\end{R}

The geometry of a GGPT does not change if we multiply $\langle \cdot,\cdot\rangle _{m,\mu}$ or, equivalently, $\langle \cdot,\cdot\rangle _{0}$ and $\mu$ by the same positive scalar. Then the parameter $\chi$ is also multiplied by the same scalar. Hence, as we see from the next propositions, the geometry of a GGPT depends only on the quotient $\chi / \mu$ being, in a sense, a dimensionless parameter of the space. Note that to some extent this quantity characterises the geometry of a GGPT and this is why we call it the \textit{space constant}. In fact, we can interpret it, at least under some additional assumptions, in the language of the maximal angle between the Bloch vectors of pure states, the orthogonal dimension of the state space, and the maximal entropy of the state. 

For the inner product of the Bloch vectors we provide an upper bound, and also a lower bound for an infra-dual GGPT (\textsl{fundamental inequalities}). In this case also the angles of the Bloch vectors between the (pure) states of
maximal norm are bounded from below. Both phenomena play the fundamental role
in QBism, but, as we see below, they are also present in its (not necessarily self-dual) GGPT counterpart.

\begin{Prop}[fundamental inequalities for states]
\label{Funine}
Let a GGPT be infra-dual and let $x,y\in B$. Then

\begin{enumerate}[i.]
\item $-\mu\leq\left\langle x_{m},y_{m}\right\rangle _{0}\leq\chi$;

\item $\cos\measuredangle(x_{m},y_{m})\geq-\mu/\chi$ for $x,y\in M(B)$ and equality holds if and only if $\langle x,y \rangle _{m,\mu} = 0$ (or, equivalently, $\langle x_{m},y_{m}\rangle _{0} = -\mu$)

\end{enumerate}
\end{Prop}

\begin{proof}
From the Cauchy-Schwartz inequality we obtain $\left\langle x_{m}
,y_{m}\right\rangle _{0}\leq\left\|  x_{m}\right\|  _{0}\left\|
y_{m}\right\|  _{0}\leq\chi$. Moreover, $\left\langle x_{m},y_{m}
\right\rangle _{0}=\left\langle x,y\right\rangle _{m,\mu}-\mu\geq-\mu$, which implies (i.). From (i.) we have $\cos\measuredangle(x_{m},y_{m})=\frac{\left\langle
x_{m},y_{m}\right\rangle _{0}}{\left\|  x_{m}\right\|  _{0}\left\|
y_{m}\right\|  _{0}}\geq -\mu/\chi$ for $x,y\in M(B)$, which gives inequality in (ii.). To get the equality in (ii.), $\left\langle x_{m},y_{m}\right\rangle _{0}=-\mu$   is necessary.
\end{proof}

\begin{R}\label{perfectlydist}
The second inequality in (i.) is also true without the assumption of a GGPT being infra-dual. If a GGPT is equinorm and self-dual we call states $x$ and $y$ fulfilling the equalities in Proposition \ref{Funine}.ii. \textsl{antipodal} \cite{IusSee05,IusSee08b,IusSee09}. It follows from Theorem \ref{compat}.iii that such pairs of pure states always exist in this case. Thus, the cosine of the \emph{maximal angle} $\psi_{\max}$ between the Bloch vectors of pure states equals $-\mu/\chi$. In particular, $\mu \leq \chi$. If a GGPT is additionally \textsl{bit symmetric} (and so necessarily self-dual and equinorm) \cite{MulUdu12}, their antipodality has also an operational interpretation, namely, they are \textsl{perfectly distinguishable}, see \cite[Lemma 3.23, Theorem 3.24]{Udu12}.
\end{R}

\begin{R}
Since the geometry on $V_0$ obviously does not depend on the parameter $\mu$, and the fundamental inequalities hold whenever the GGPT in question is infra-dual, one can easily conclude that the parameter $\mu$ in Proposition \ref{Funine} can be replaced by $\mu_i$.
\end{R}
Few other interpretations of the parameter $\chi/\mu$, together with an explanation, why under some assumptions $m$ can be referred to as the \emph{maximally mixed state}, are given in the Appendix \ref{Appendix constant}.

\subsection{Self-duality}
\label{Self-duality}
There are plenty of sufficient and necessary conditions for $(V,C)$ or for GPT $(V,C,e)$ to be self-dual, see e.g. \cite{BarFor76,Ioc84,Khu98,MulUdu12,Udu12,Wil12,JanLal13,Baretal14,Wil18,ItoLou19}. In this section we propose another characterisation of this notion.

Namely, from Theorem \ref{compat}.\ref{sd}. and Remark \ref{supra} we obtain straightforwardly the following result. For $e\in\operatorname*{int}  C^\ast$ we shall write $B_e:=\{x\in C:e(x)=1\}$, $V^e_0:=e^{-1}(0)$, and $V^e_1:=e^{-1}(1)$ to stress that all these sets actually depend on $e$.
\begin{Th}
\label{Self-dual equivalent}
The following conditions are equivalent:
\begin{enumerate}[i.]
\item $(V,C)$ is self-dual;
\item for every $e\in\operatorname*{int}  C^\ast$ there are an inner product $\langle\cdot,\cdot\rangle_0$ on $V^e_0$ and a state $m\in \operatorname*{int} _{V^e_1}B_e$ such that $\min_{y\in\operatorname*{ex}B_e}\left\langle
x_{m},y_{m}\right\rangle _{0}=\operatorname*{const}(x\in\partial B_e)$;
\item there are $e\in\operatorname*{int}  C^*$, an inner product $\langle\cdot,\cdot\rangle_0$ on $V^e_0$ and a state $m\in \operatorname*{int} _{V^e_1}B_e$ with $\min_{y\in\operatorname*{ex}B_e}\left\langle
x_{m},y_{m}\right\rangle _{0} \linebreak =\operatorname*{const}(x\in\partial B_e)$.
\end{enumerate}
\end{Th}

This means that if the self-duality of $(V,C)$ is realised (in the sense of a GGPT) for a given base $B$ of $C$ (or, equivalently, for a given unit effect $e \in \operatorname*{int}  C^\ast$), then it is realised for any other base. This is also visible if we consider another approach to self-duality.

\begin{Df}
Let $A$ be a convex subset of an affine subspace $\mathcal A$ of a Euclidean space $V$ endowed with the inner product $\langle \cdot,\cdot\rangle$. 
Define the \textsl{polar} $A^\circ_{c,s}$ (or $A^\circ$ for short) and the \textsl{dual} $A^\star_{c,s}$ (or $A^\star$ for short) of $A$ in $\mathcal A$ \textsl{with respect to the sphere with centre at} $c\in\mathcal A$ \textsl{of radius} $s>0$ as 
$A^\circ:=\{x\in\mathcal A:s^2\geq\langle x-c,y-c\rangle \textnormal{ for every }y\in A\}$ and $A^\star:=\iota(A^\circ)$,
where $\iota$ is the \textsl{inversion} in $\mathcal A$ through $c$  given by $\iota(x):=2c-x$, $x\in\mathcal A$. Note that $A^\star=\{x\in\mathcal A:\langle x-c,y-c\rangle\geq-s^2 \textnormal{ for all } y\in A\}$. We say that $A$ is \textsl{self-polar} (resp. \textsl{self-dual}) if $A^\circ=A$ (resp. $A^\star=A$). Moreover, polar and dual sets of $A$ are related in the following manner: $A^{\circ\circ}=A$, $A^{\star}=\left(  \iota\left( A\right)  \right)  ^{\circ}$, and $A^{\star\star}=A$.
\end{Df}

From Proposition \ref{ProInn}.\ref{inn0} we deduce that if $B$ is a base for $C$, then $B^\star_{m,\sqrt{\mu}}$ is a base of the cone $C^{+}$. Hence it follows that the self-duality of a GGPT is equivalent to the self duality of $B$ with respect to the sphere with centre at $m$ of radius $\sqrt{\mu}$. From this fact we get easily the following equivalent conditions for self-duality of $(V,C)$.

\begin{Th}
\label{selfdual}
The following conditions are equivalent:
\begin{enumerate}[i.]
\item $(V,C)$ is self-dual;
\item for every $e\in \operatorname*{int} C^\ast$ there exist an inner product $\langle\cdot,\cdot\rangle_0$ on $V^e_0$, $m\in \operatorname*{int} _{V^e_1}B_e$, and $s>0$ such that $B_e$ is self-dual in $V^e_1$ with respect to $B(m,s)$;
\item there exist $e\in \operatorname*{int}  C^\ast$, an inner product $\langle\cdot,\cdot\rangle_0$ on $V^e_0$, $m\in \operatorname*{int} _{V^e_1}B_e$, and $s>0$ such that $B_e$ is self-dual in $V^e_1$ with respect to $B(m,s)$.
\end{enumerate}
Moreover, it follows from Theorem \ref{compat}.\ref{sd} that in the above statements $s$ is unique and $s^2$ equals $\mu=-\min_{y\in\operatorname*{ex}B_e}\left\langle
x_{m},y_{m}\right\rangle _{0}$ for every $x\in\partial B_e$.
\end{Th}

\begin{R}
\label{innout}
If a GGPT is self-dual, then it follows from Proposition \ref{Funine} and Theorem \ref{selfdual} that $B_e$ lies between two dual sets with respect to $B(m,\sqrt{\mu})$, `inner' and `outer ball': $B(m,r) \subset B_e \subset B(m,R)$, where $r:=\mu /\sqrt {\chi}$, $R:=\sqrt {\chi}$, and $rR=\mu$, $R/r=\chi/\mu$. Thus, for self-dual GGPTs  the coefficient $\chi/\mu \geq 1$ is also related to the \textsl{asphericity} of $B_e$ \cite{Aubetal19}.
\end{R}

The ball and the regular-simplex are two extreme (in terms of coefficient $\chi/\mu$) examples of self-dual  GGPTs endowed with an \emph{orthogonal frame}, i.e. a set $\Omega$ of mutually orthogonal (in the sense of $\langle\cdot,\cdot\rangle_{m,\mu}$) elements from $M(B)$ (or, equivalently, a set of perfectly distinguishable states, see Remark \ref{perfectlydist}) such that $m\in\operatorname{aff}(\Omega)$. Proof of this fact is presented in Appendix \ref{Appendix constant}.

\begin{Prop}[extreme cases] 
\label{extreme}
Let a GGPT be self-dual and endowed with an orthogonal frame. Then
\begin{equation}
1 \leq \chi/\mu \leq \dim{V_0}.
\end{equation}
Moreover,
\begin{enumerate}[i.]
    \item $\chi/\mu=1$ if and only if $B$ is a ball; see Ex. (D) below;
    \item $\chi/\mu=\dim{V_0}$ if and only if $B$ is a regular simplex; see Ex. (B) below. 
\end{enumerate}
\end{Prop}

\subsection{Examples of GGPTs}
\label{exGGPT}
The following examples of GGPTs are in fact well known and will be used in this paper.

\begin{enumerate}[A.]

\item (\textsl{classical}) The \textsl{classical GGPT} is given by $V:=\mathbb{R}^{N}$,
$C:=\{x\in\mathbb{R}^{N}:x_{j}\geq0,j=1,\ldots,N\}$, $e(x):=\sum
\nolimits_{j=1}^{N}x_{j}$, and $\left\langle x,y\right\rangle :=\sum
\nolimits_{j=1}^{N}x_{j}y_{j}$ for $x,y\in\mathbb{R}^{N}$. Then $V_{0}=\{x\in\mathbb{R}^{N}:\sum\nolimits_{j=1}^{N}x_{j}=0\}$, $V_{1}=\{x\in
\mathbb{R}^{N}:\sum\nolimits_{j=1}^{N}x_{j}=1\}$, and the states form the
probability simplex $B=\Delta_{N}$ with the center at $m=(1/N,\ldots1/N)$. Note
that pure states are vertices of $\Delta_{N}$. The parameters of the space are
$\mu=1/N$ and $\chi=1-1/N$ with $\chi/\mu=N-1$. This GGPT is spectral, equinorm and self-dual, with $\psi_{\max} =\arccos \frac{1}{1-N}$.

\item  (\textsl{quantum}) Let $\mathfrak{H}$ be a complex finite-dimensional Hilbert space,
$\mathfrak{H}\simeq\mathbb{C}^{d}$ ($d\in\mathbb{N}$). Then the \textsl{quantum GGPT} is defined as: $V$ -- the space of linear self-adjoint operators on
$\mathfrak{H}$ ($\dim V = d^2$), $C$ -- positive elements in $V$, $e(\rho):=\operatorname*{tr}\rho$
for $\rho\in V$ (the \textsl{trace functional}), and $\left\langle \rho
,\sigma\right\rangle :=\operatorname*{tr}\rho\sigma$ for $\rho,\sigma\in V$
(the \textsl{Hilbert-Schmidt inner product}). Then the states $B=\{\rho\in
V:\rho\geq0,\operatorname*{tr}\rho=1\}$ are \textsl{density operators}, with
the centre of this set at $m=I/d$, and pure states $\operatorname*{ex}B$ can by identified
with the projective complex space $\mathbb{C}P^{d-1}$. The parameters of the
space are $\mu=1/d$ and $\chi=1-1/d$ with $\chi/\mu=d-1$. This GGPT is spectral, equinorm and
self-dual, with $\psi_{\max}=\arccos\frac{1}{1-d}$. In particular, for
$d=2$ the space is isomorphic with the closed unit ball in $\mathbb{R}^{3}$
(the \textsl{Bloch ball}) and the set of pure states is isomorphic with its
boundary, i.e. with the two-dimensional unit sphere (the \textsl{Bloch}
\textsl{sphere}) with $\psi_{\max}=\pi$.

\item\label{poly} (\textsl{regular polygonal}) Take $N\in\mathbb{N}$, $N\geq3$. Let $u_{l}:=(\cos(2\pi l/N),\sin(2\pi
l/N))\in\mathbb{R}^{2}$, $l=1,\ldots,N$ be the vertices of a regular $N$-gon in $\mathbb{R}^{2}$, such that $\left\|  u_{l}\right\|  _{2}=1$ and
$\sum\nolimits_{l=1}^{N}u_{l}=0$. Define the \textsl{regular polygonal GGPT} by
$V:=\mathbb{R}^{3}$, $C:=\{\sum\nolimits_{l=1}^{N}t_{l}z_l  :z_l:=\left(  au_{l}
,1\right),t_{l}\geq0,l=1,\ldots,N\}$, $a>0$, $e(x):=x_{3}$, and
$\left\langle x,y\right\rangle :=\sum\nolimits_{j=1}^{3}x_{j}y_{j}$ for $x,y \in\mathbb{R}^3$. Then
$m=(0,0,1)$ and $\mu=1$. Moreover, $B=\{\sum\nolimits_{l=1}^{N}t_{l}
z_{l}:\sum\nolimits_{l=1}^{N}t_{l}=1,t_{l}\geq0,l=1,\ldots,N\}$ and
$\operatorname*{ex}B=\{z_{l}:l=1,\ldots,N\}$. Thus the GGPT is equinorm with $\chi=a^{2}$. 

Let $l,s=1,\ldots,N$. Then we have $\left\langle (au_{l}
,1),(au_{s},1)\right\rangle =a^{2}\left\langle u_{l},u_{s}\right\rangle
+1=a^{2}\cos(2\pi(l-s)/N)$. Now, from Theorem \ref{compat}.\ref{id} it follows immediately that to make the GGPT infra-dual the condition $a^{2}\min_{l,s=1,\ldots,N}\cos(2\pi(l-s)/N)\geq
-\mu=-1$ must be fulfilled. Hence, for even $N$ we get $0<a\leq1$ and for odd $N$ we obtain $0<a \leq (\cos(\pi/N))^{-1/2}$. On the other hand, from Theorem \ref{compat}.\ref{suprad} we get that for even $N$ supra-duality of the GGPT is equivalent to $a \geq (\cos(\pi/N))^{-1} > 1$, and for odd $N$ to $a \geq (\cos(\pi/N))^{-1/2}$. Applying Theorem \ref{compat}.\ref{sd} we deduce that for $N$ odd, choosing optimal $a=(\cos(\pi/N))^{-1/2}$, we get a self-dual
GGPT with $\psi_{\max}=\pi(1-1/N)$, and, on the other hand, such choice is not
possible for $N$ even and any $a>0$. Note that the above inequalities provide also bounds on $\chi/\mu$ since in this case $\chi/\mu=a^2$.

In particular, for $N$ odd we get $\chi/\mu=(\cos(\pi/N))^{-1}$ in the self-dual case, e.g.
$\chi/\mu=2$ and $\psi_{\max}=2\pi/3$ for $N=3$ (this is just the classical triangular GGPT), and $\chi/\mu=\sqrt{5}-1$ and $\psi_{\max}=4\pi/5$
for $N=5$ (the \textsl{pentagonal} GGPT analysed in details in Example \ref{penta}). For $N=4$ we get a non self-dual \textsl{gbit} or \textsl{Boxworld} GGPT \cite{Bar07,Janetal11,Pla21}. In the limit $N \rightarrow \infty$ we get the (self-dual) ball GGPT, see the next example.

\item\label{exball} (\textsl{ball}) Set $N\in\mathbb{N}$. For the \textsl{ball} GGPT we put 
$V:=\mathbb{R}^{N}\oplus\mathbb{R}$,
$C:=\{(x,\lambda):x\in\mathbb{R}^{N},\lambda\in\mathbb{R},\left\|  x\right\|
_{2}\leq  \lambda \}$, $e(x,\lambda):=\lambda$ for
$(x,\lambda)\in C$, and $\left\langle (x,\lambda),(y,\kappa)\right\rangle
:=\sum\nolimits_{j=1}^{N}x_{j}y_{j}+\lambda\kappa$ for $(x,\lambda
),(y,\kappa)\in V$. The set of states is then the $N$-dimensional unit
\textsl{ball} lifted upwards, $B=\{(x,1):x\in\mathbb{R}^{N},\left\|  x\right\|  _{2}\leq1\}$,
and $\operatorname*{ex}B=\{(x,1):x\in\mathbb{R}^{N},\left\|  x\right\|
_{2}=1\}$ is the $(N-1)$-dimensional unit sphere. Moreover, $m=(0,\ldots,0,1)$
and $\mu=1$. These GGPTs are spectral, equinorm with $\chi=\chi/\mu=1$ and self-dual.
Clearly, $\psi_{\max}=\pi$. The cone $C$ is also known as the \textsl{Lorentz `ice cream' cone}.
\end{enumerate}

The Koecher-Vinberg theorem \cite{Koe57,Vin67,Wil18,Baretal20} says that a GPT given by an abstract state space $(V,C,e)$ is self-dual and \textsl{homogeneous} (i.e. such that the group of linear automorphisms of $V$ transforming $C$ onto itself acts transitively on $\operatorname*{int} C$) if and only if $(V,C)$ is order-isomorphic to a \textsl{formally real (or Euclidean) Jordan algebra}. The examples (B) and (D) are of this type.

\section{Morphophoricity}
\label{Morphophoricity}

As already mentioned in the introduction, in order to define morphophoricity, we need to assume that the set of states $B$ is equipped with some Euclidean geometry, i.e. we can define not only the distances between states but also the angles between the respective vectors in the underlying vector space $V_0$. This naturally leads to introducing some inner product in $V_0$. However, at this point we do not require the full state space $V$ to be an inner product space -- it suffices that the subspace $V_0$ is supplied with an inner product. We shall be interested in the properties of measurements which enable measurement maps to transfer the geometry of the set of states to the probability simplex.

Let us denote by $\langle\cdot,\cdot\rangle_0$ an inner product on $V_0$ and by $\|\cdot\|_0$ the respective norm in $V_0$.

\begin{Df}

We say that a measurement $\pi$ is \textsl{morphophoric} (\textsl{with respect to} $\langle\cdot,\cdot\rangle_0$) if there exists $\alpha > 0$ such that
\begin{equation}
\|\pi(x)-\pi(y)\|^2=\alpha\|x-y\|_0^2 \quad \text{for  } x,y\in B.
\end{equation}
In other words, $\pi$ is a similarity of $B$ and $\pi(B) \subset \Delta_n$ with the similarity ratio $s:=\sqrt\alpha$. The definition can be equivalently written as \begin{equation}\langle\pi(x)-\pi(z),\pi(y)-\pi(z)\rangle =\alpha\langle x-z,y-z\rangle_0 \quad \text{for  } x,y,z\in B.
\end{equation}
In fact, $\pi$ acts also as a similarity on $V_0$
\begin{equation}\langle\pi(x),\pi(y)\rangle =\alpha\langle x,y\rangle_0 \quad \text{for  } x,y\in V_0,
\end{equation}
but not necessarily on the whole space $V$.
\end{Df}
The morphophoricity of a measurement implies in a obvious way its informational completeness.

\begin{Df}
	We say that $\pi$ is \textsl{informationally complete} if $\pi(x)=\pi(y)$ implies $x=y$ for all $x,y\in B$.	
\end{Df}

Thus, the statistics of the informationally complete measurement outcomes uniquely determine the pre-measurement state. Another simple characterisation of informational completeness is provided by the following result, which in particular implies that $n\geq\dim V$.

\begin{Th}\cite[Lemma 2.1]{SinStu92}\label{IC1}
The conditions below are equivalent:
\begin{enumerate}[i.]
	\item $\pi$ is informationally complete,
	\item $\lin\{\pi_j:j=1,\ldots,n\}=V^*$.
\end{enumerate}
\end{Th}

If we assume that a distinguished `state' $m\in V_{1}$ and a
size parameter $\mu>0$ are given, and generate the inner product in $V$ by \eqref{innpro}, then  the dual space $V^*$ is also equipped with the uniquely defined induced inner product, denoted also by $\langle\cdot,\cdot\rangle_{m,\mu}$. 
Let $V_0^*:=\{f\in V^*:f(m)=0\}$ (not to be confused with $(V_0)^*$). Then the orthogonal projections onto $V_0$ and $V_0^*$ are given by $P_0:V\ni x\mapsto x-e(x)m\in V_0$ and $\mathcal{P}_0:V^*\ni f\mapsto f-f(m)e\in V^*_0$, respectively. Note that $\mathcal P_0(T_{m,\mu}(x))=T_{m,\mu}(P_0(x))$ for $x\in V$, and $T_{m,\mu}^{-1}(\mathcal P_0(f))=P_0(T_{m,\mu}^{-1}(f))$ for $f\in V^*$. Now we can easily extend Theorem \ref{IC1}.
\begin{Th}
\label{IC}
The following conditions are equivalent:
	\begin{enumerate}[i.]
		\item\label{ICbas} $\pi$ is informationally complete,
		\item $\lin\{\pi_j:j=1,\ldots,n\}=V^*$,
		\item $\lin\{\mathcal{P}_0(\pi_j):j=1,\ldots,n\}=V_0^*$,
		\item $\lin\{T_{m,\mu}^{-1}(\pi_j):j=1,\ldots,n\}=V$,
		\item\label{ICV0}$\lin\{P_0(T_{m,\mu}^{-1}(\pi_j)):j=1,\ldots,n\}=V_0$.
	\end{enumerate}
\end{Th}

The existence of an inner product on $V$ allows us to characterise morphophoricity in terms of tight frames. Below (up to Theorem \ref{traceformula}) we provide a brief reminder of the basic facts  concerning these objects, see e.g. \cite{Wal17}. 
 Let $\mathcal H$ be a finite-dimensional Hilbert space with inner product $\langle\cdot|\cdot\rangle$ and let
$H:=\{h_1,\ldots,h_m\}\subset \mathcal H$. 
The operator $S:=\sum_{i=1}^m|h_i\rangle\langle h_i|$ is called the \textit{frame operator}.
\begin{Df}	
$H$ is a \textit{tight frame} if there exists $A>0$ such that \begin{equation}
A\|v\|^2=\sum_{i=1}^m|\langle v|h_i\rangle|^2\quad \textnormal{for all } v\in \mathcal H.
\end{equation}
	In such situation $A$ is referred to as the \textit{frame bound}. 
\end{Df}

The next theorem provides some equivalent conditions for a set of vectors to be a tight frame that justify why tight frames can be thought of as some generalisations of the orthonormal bases.

\begin{Th}
\label{tf}
Let $H:=\{h_1,\ldots,h_m\}\subset\mathcal H$ and $A>0$.
The following conditions are equivalent:
	\begin{enumerate}[i.]
		\item $H$ is a tight frame with the frame bound $A$,
		\item\label{tf2} $\sum_{i=1}^m\langle u|h_i\rangle\langle h_i|v\rangle=A \langle u|v\rangle$ for all $u,v\in \mathcal H$,
		\item $Av=\sum_{i=1}^m \langle h_i|v\rangle h_i$ for every $v\in \mathcal H$.
	\end{enumerate}
\end{Th}
The following is a useful formula that allows one to express the frame bound in terms of the space dimension and norms of the frame elements.
\begin{Th}[trace formula]\label{traceformula}
	Let $H$ be a tight frame. Then	
\begin{equation}
\label{trace}
	A=\frac{1}{\dim\mathcal H}\sum_{i=1}^m\|h_i\|^2.
\end{equation}
\end{Th}

We can finally characterise morphophoric measurements in terms of tight frames.
\begin{Th}
\label{morphTF} 
Let $(V,C,e)$ be an abstract state space, $\langle\cdot,\cdot\rangle_0$ -- an inner product on $V_0$ and $\pi:V\to\mathbb R^n$ -- a measurement. The following conditions are equivalent:
\begin{enumerate}[i.]
	\item the measurement $\pi$ is morphophoric (with respect to $\langle\cdot,\cdot\rangle_0$),
	\item for every extension of $\langle\cdot,\cdot\rangle_0$ onto $V$ via some $m\in V_1$ and $\mu>0$, $(P_0(T_{m,\mu}^{-1}(\pi_j)))_{j=1}^n$ is a tight frame for $V_0$, 
	\item for every extension of $\langle\cdot,\cdot\rangle_0$ onto $V$ via some $m\in V_1$ and $\mu>0$ and the induced inner product on~$V^*$, $(\mathcal P_0(\pi_j))_{j=1}^n$ is a tight frame for $V_0^*$,
	\item there exists an extension of $\langle\cdot,\cdot\rangle_0$ onto $V$ via some $m\in V_1$ and $\mu>0$ such that $(P_0(T_{m,\mu}^{-1}(\pi_j)))_{j=1}^n$ is a tight frame for $V_0$,
	\item\label{existsstar} there exists an extension of $\langle\cdot,\cdot\rangle_0$ onto $V$ via some $m\in V_1$ and $\mu>0$ and the induced inner product on $V^*$, such that $(\mathcal P_0(\pi_j))_{j=1}^n$ is a tight frame for $V_0^*$.
	\end{enumerate}  
Moreover, the parameter $\alpha$ of the morphophoricity coincides with with the frame bound and is given by 
		\begin{equation}\label{alpha1}
		    \alpha=\frac{1}{\dim V_0}\sum_{j=1}^n\left(\|T_{m,\mu}^{-1}(\pi_j)\|_{m,\mu}^2-\mu(e(T_{m,\mu}^{-1}(\pi_j)))^2\right)=\frac{1}{\dim V_0}\sum_{j=1}^n (\|\pi_j\|_{m,\mu}^2-\frac{1}{\mu}(\pi_j(m))^2).
		\end{equation}	
\end{Th}

\begin{proof} The implications (ii.)$\,\Rightarrow\,$(iv.) and (iii.)$\,\Rightarrow\,$(v.) are obvious.
	
	 (iv.)$\,\Rightarrow\,$(ii.) Let $m,m'\in V_1$,  $\mu,\mu'>0$ and $P_0,P_0':V\to V_0$ -- the corresponding orthogonal projections onto $V_0$. Let $f\in V^*$ and $w\in V_0$. Then $$\langle P_0(T_{m,\mu}^{-1}(f)),w\rangle_0=\langle T_{m,\mu}^{-1}(f), w\rangle_{m,\mu}=f(w)=\langle T_{m',\mu'}^{-1}(f),w\rangle_{m',\mu'}=\langle P'_0(T_{m',\mu'}^{-1}(f)),w\rangle_0,$$
	and so $P_0(T_{m,\mu}^{-1}(f))=P_0'(T_{m',\mu'}^{-1}(f))$ for any $f\in V^*$. Thus if $(P_0(T_{m,\mu}^{-1}(\pi_j)))_{j=1}^n$ is a tight frame in $V_0$ for one extension of $\langle\cdot,\cdot\rangle_0$, it is a tight frame for every such extension. Moreover, from \eqref{trace} it follows that the parameter $\alpha$ does not depend on the choice of $m$ and $\mu$ as it is determined by the trace formula $$\alpha=\frac{1}{\dim V_0}\sum_{j=1}^n\|P_0(T_{m,\mu}^{-1}(\pi_j))\|^2_{m,\mu}=\frac{1}{\dim V_0}\sum_{j=1}^n\left(\|T_{m,\mu}^{-1}(\pi_j)\|^2_{m,\mu}-\mu(e(T_{m,\mu}^{-1}(\pi_j)))^2\right).$$
	The equivalences between (ii.) and (iii.), and between (iv.) and (v.) follow from the fact that $V$ and $V^*$ are isometrically isomorphic via $T_{m,\mu}$, and, by Proposition \ref{ProInn}.\ref{minn}., we obtain $e(T_{m,\mu}^{-1}(\pi_j))=\mu^{-1}\langle T_{m,\mu}^{-1}(\pi_j),m\rangle_{m,\mu}=\mu^{-1}\langle\pi_j,T_{m,\mu}(m)\rangle_{m,\mu}=\langle\pi_j,e\rangle_{m,\mu}=\mu^{-1}\pi_j(m)$. 
	
	(i.)$\,\Leftrightarrow\,$(ii.) Let $m\in V_1$, $\mu>0$ and $v,w\in V_0$.
Then 
\begin{align*}
	\langle \pi(v),\pi(w)\rangle &=\sum_{j=1}^n\pi_j(v)\pi_j(w)
	\\
	&=\sum_{j=1}^n\langle v,T^{-1}_{m,\mu}(\pi_j)\rangle_{m,\mu}\langle T^{-1}_{m,\mu}(\pi_j),w\rangle_{m,\mu}\\&= \sum_{j=1}^n\langle v,P_0(T^{-1}_{m,\mu}(\pi_j))\rangle_{m,\mu}\langle P_0(T^{-1}_{m,\mu}(\pi_j)),w\rangle_{m,\mu}.	
\end{align*}	
Thus, by Theorem \ref{tf}, $\pi$ is morphophoric (with respect to $\langle\cdot,\cdot\rangle_0$) with the similarity ratio $\sqrt\alpha$ if and only if $(P_0(T^{-1}_{m,\mu}(\pi_j)))_{j=1}^n$ is a tight frame in $V_0$ with the frame bound $\alpha$. As above, $\alpha$ can be calculated  from the trace formula.
\end{proof}

Note that the set $(P_0(T_{m,\mu}^{-1}(\pi_j)))_{j=1}^n$ is always \textit{balanced}, i.e.
\begin{equation}
\sum_{j=1}^n P_0(T_{m,\mu}^{-1}(\pi_j))=P_0(T_{m,\mu}^{-1}(e))=0.
\end{equation}

Morphophoric measurements are not rare; in fact, as we shall see from the construction below, this class is wide for each GGPT. The proof follows the same idea that is used in \cite{App07} to show the existence of some general symmetric quantum measurements.

\begin{Th}
\label{existence}
For every GPT $(V,C,e)$ equipped with an inner product $\langle\cdot,\cdot\rangle_0$ on $V_0$ (so, in particular, for any GGPT) there exists a morphophoric measurement.
\end{Th}

\begin{proof}
Let us consider any extension of $\langle\cdot,\cdot\rangle_0$ onto $V$ via some $m\in V_1$ and $\mu>0$, as well as the induced inner product on $V^*$. Since $e\in\operatorname{int}C^*$, there exists $r>0$ such that $\overline{B}(e,r)\subset C^*$. Obviously, $\mathcal P_0\left(\overline{B}(e,r)\right)=\overline{B}_{V_0^*}(0,r)$. Let $(u_j)_{j=1}^{\dim V}$ be the vertices of a regular simplex inscribed in $\overline{B}_{V_0^*}(0,r)$ (and therefore a balanced tight frame for $V_0^*$), and let $\pi_j:=\frac{1}{\dim V}\left(u_j+e\right)$ for $j=1,\ldots,\dim V$. Since $u_j+e\in \overline{B}(e,r)$,  we get $\pi_j\in C^*$. Moreover, $\sum_{j=1}^{\dim V}\pi_j=\frac{1}{\dim V}\left(\sum_{j=1}^{\dim V}u_j+\dim V\cdot e\right)=e$, and so $(\pi_j)_{j=1}^{\dim V}$ is a measurement. Its morphophoricity follows from the fact that  $\mathcal P_0(\pi_j)=\frac{1}{\dim V}u_j$ and the condition (\ref{existsstar}.) in Theorem \ref{morphTF}.
\end{proof}

\begin{R}
\label{existence2}
The construction of a morphophoric measurement proposed in the proof above is quite specific, however, it is easy to see how any such measurement can be obtained. Indeed, instead of a regular simplex, let $(u_j)_{j=1}^n\subset\overline{B}_{V_0^*}(0,r)$ be \emph{any} balanced tight frame for $V_0^*$, and let $(\alpha_j)_{j=1}^n$ be any sequence of positive numbers such that $u_j+\alpha_j e\in C^*$. Then $((\sum_{j=1}^n\alpha_j)^{-1}(u_j+\alpha_j e))_{j=1}^n$ is a morphophoric measurement.   
\end{R}

The next theorems generalise the results from \cite{SloSzy20}. 
Firstly, we collect simple facts that give us ways to obtain new morphophoric measurements from the given ones. We denote the \emph{white noise} (understood as a measurement giving some random answers independently of the initial state) by $qe=(q_je)_{j=1}^n$ for some $q=(q_1,\ldots,q_n)\in\Delta_n$.
\begin{Prop}
\label{str1}
	\begin{enumerate}[i.]
		\item  If $\pi^1,\ldots,\pi^m$ are morphophoric measurements with squares of the similarity ratios equal to $\alpha_1,\ldots,\alpha_m$, then also $\pi := (t_1\pi^1)\cup\ldots\cup(t_m\pi^m)$ is a morphophoric measurement for any $t_1,\ldots,t_m\geq0$ such that $t_1+\ldots+t_m= 1$.  In such case the square of the similarity ratio for $\pi$ is equal to $\alpha=t_1\alpha_1+\ldots+t_m\alpha_m$.
		\item Let $\pi =  (\pi_j)^n_{j=1}$ be  a  morphophoric  measurement  and  let $q\in\Delta_n$.   Then $\pi_{\lambda,q}:=\lambda\pi + (1-\lambda)qe$ is  also  a  morphophoric  measurement  for $\lambda \in(0,1]$.   In  such  case,  the  square  of  the similarity ratio for $\pi_{\lambda,q}$ is equal to $\lambda^2\alpha$, where $\alpha$ is the square of the similarity ratio for $\pi$.
	\end{enumerate}
\end{Prop}
\begin{proof}
	\begin{enumerate}[i.]
		\item Follows directly from the definition of tight frame.
		\item It is enough to observe that $\pi_{\lambda,q}(B)=\lambda\pi(B)+(1-\lambda)q$ which is a homothety (therefore a similarity) with the centre at $q$ and the ratio equal to $\lambda$. \qedhere
	\end{enumerate}	
\end{proof}

\begin{Df}
	We say that the measurement $\pi=(\pi_j)_{j=1}^n$ is \emph{boundary} if $\pi_j$ is boundary for every $j=1,\ldots,n$, i.e. for every $j=1,\ldots,n$ there exists $x\in B$ such that $\pi_j(x)=0$.	
\end{Df}
The next two theorems tell us that with any morphophoric measurement one can associate two special boundary morphophoric measurements.

\begin{Th}
\label{str2}
	Let $\pi$ be a morphophoric measurement which is not boundary.  Then there exists a unique boundary  morphophoric  measurement $\sigma= (\sigma_j)^n_{j=1}$ such  that $\pi =\lambda\sigma+ (1-\lambda)qe$ for  some $\lambda\in(0,1)$ and $q\in\Delta_n$.
	
\end{Th}

\begin{proof} Let $\lambda_j:=\min_{x\in B}\pi_j(x)$ for $j=1,\ldots,n$. There exists $j_0\in\{1,\ldots,n\}$ such that  $\pi_{j_0}(x)>0$ for every $x\in B$ and thus $\sum_{j=1}^n\lambda_j>0$. On the other hand, $$\sum_{j=1}^n\lambda_j\leq\min_{x\in B}\sum_{j=1}^n\pi_j(x)=\min_{x\in B}e(x)=1$$ and the equality holds if and only if  there exists $u\in B$ such that $\pi_j(u)=\lambda_j$ for every $j=1,\ldots,n$. But, by the morphophoricity, $\pi\neq pe$ for any $p\in\Delta_n$. Thus, if such $u$ exists, then there exist $j\in\{1,\ldots,n\}$ and $v\neq u$ such that $\pi_j(v)>\pi_j(u)$. But in such case $1=e(v)=\sum_{j=1}^n\pi_j(v)>\sum_{j=1}^n\pi_j(u)=e(u)=1$, a contradiction. In consequence, $\sum_{j=1}^n\lambda_j<1$. Put $$q_j:=\lambda_j/\sum_{j=1}^n\lambda_j\in[0,1],\ \lambda:=1-\sum_{j=1}^n\lambda_j\in(0,1)\textnormal{ and }\sigma_j:=(\pi_j+(\lambda-1)q_je)/\lambda.$$ Then $\min_{x\in B}\sigma_j(x)=(\lambda_j+(\lambda-1)q_j)/\lambda=0$. Clearly, $\sum_{j=1}^n\sigma_j=e$. Thus $\sigma=(\sigma_j)_{j=1}^n$ is a boundary measurement. The morphophoricity of $\sigma$ follows from the fact that $\mathcal{P}_0(\sigma_j)=\mathcal{P}_0(\pi_j)/\lambda$.
	
	Now, to see the uniqueness, let  $\sigma'$ be a boundary morphophoric measurement such that $\pi=\lambda'\sigma'+(1-\lambda')q'e$ for some $\lambda'\in (0,1)$ and $q'\in\Delta_n$. Then $0=\min_{x\in B}\sigma_j'(x)=\frac{1}{\lambda_j'}(\lambda_j+(\lambda'-1)q_j')$. Thus for every $j=1,\ldots,n$ we obtain $\lambda_j=(1-\lambda')q_j'=(1-\lambda)q_j$ and in consequence by summing over $j$ we get $1-\lambda'=1-\lambda$. Therefore $\lambda'=\lambda$ and $q'=q$. \qedhere	
\end{proof}

\begin{Th}
\label{str3}
	Let $\pi$ be  a  morphophoric  measurement.  Then  there  exists  a  unique  boundary  morphophoric measurement $\tilde{\sigma}$ such that $\lambda\pi + (1-\lambda)\tilde{\sigma}=qe$ for some $\lambda\in(0,1)$ and $q\in\Delta_n$.
	
\end{Th}

\begin{proof}
	Let $\mu_j:=\max_{x\in B}\pi_j(x)$ for $j=1,\ldots,n$. In a similar way as before we show that $\sum_{j=1}^n\mu_j>1$. Put $$q_j:=\mu_j/\sum_{j=1}^n\mu_j\in[0,1],\  \lambda:=1/\sum_{j=1}^n\mu_j\textnormal{ and }\tilde{\sigma_j}:=(1-\lambda)^{-1}(q_je-\lambda\pi_j).$$ Then $\min_{x\in B}\tilde{\sigma_j}(x)=(1-\lambda)^{-1}(-\lambda\mu_j+q_j)=0$. Thus $\tilde{\sigma}=(\tilde{\sigma_j})_{j=1}^n$ is a boundary measurement. The morphophoricity of $\tilde{\sigma}$ follows from the fact that $\mathcal{P}_0(\tilde{\sigma_j})=(\lambda/(1-\lambda))\mathcal{P}_0(\pi_j)$.
	
	The uniqueness follows similarly as in the proof of the previous theorem.
\end{proof}

\section{Geometry of the generalised qplex}
\label{Geometry of the generalised qplex}

\subsection{General case}
\label{General case}

In this section we take a closer look at the geometry of the \textsl{generalised qplex} $\mathcal P:=\pi(B)$ when a measurement $\pi$ is morphophoric. Obviously, the internal geometry of this set is the same as the internal geometry of $B$, as they are similar. Thus, we are interested in its external geometry, in particular, where it is located in the probability simplex $\Delta_n$. Note that some of the notions and properties investigated in this section are strictly connected with the particular choice of $m\in\textnormal{int}_{V_1} B$ and $\mu>0$ used to extend the inner product $\langle\cdot,\cdot\rangle_0$ from $V_0$ to the whole space $V$, while the others depend solely on the geometry of $V_0$.

The basic image that emerges from the QBism approach to quantum theory \cite{Appetal17} is that the (Hilbert) qplex, i.e. the image of the $d$-dimensional quantum state space by a SIC-POVM measurement, is sandwiched between two dual simplices: the probability simplex and the so-called `basis' simplex. In our previous paper \cite{SloSzy20} we presented a generalisation of this property to any morphophoric quantum measurement by replacing the simplices with two dual polytopes  lying in a  $(d^2 - 1)$-section of the probability simplex by an affine space. It turns out that this observation is not quantum-specific. Indeed, we show in this section that it holds for any GGPT, even not necessarily self-dual, as the supra-duality is enough. However, the self-duality enriches this image.

 Let now introduce some definitions and notation. By the \textsl{primal affine space} we mean the affine span of the image of the set of states by the measurement $\pi$, and denote it by $\mathcal A$, i.e. $\mathcal A:=\operatorname{aff}(\pi(B))$. The corresponding linear subspace of $\mathbb{R}^n$ is denoted by $L:=\mathcal A-\mathcal A=\pi(V_0)$. 
 
 For $j=1,\dots,n$ we introduce vectors $v_j:=T_{m,\mu}^{-1}(\pi_j)\in V$ and their normalised versions $w_j:=v_j/e(v_j)\in V_1$. Note that $e(v_j)=\pi_j(m)/\mu>0$ for $m\in\textnormal{int}_{V_1}B$, and thus $w_j$ are well defined. Then $\langle v_j,x\rangle_{m,\mu}=\pi_j(x)$ for $x \in V$. Moreover, $\sum_{j=1}^n v_j = m/\mu$. The images by  $\pi$ of $w_j$ are denoted by $f_j:=\pi(w_j)\in\mathcal A$ and called \textsl{basis quasi-distributions}. We also distinguish the \textsl{central quasi-distribution} $c:=\pi(m)\in\mathcal A$ with its coordinates given by $c_j=\pi_j(m)=\langle v_j,m\rangle_{m,\mu}=\mu e(v_j)=e(v_j)/(\sum_{l=1}^ne(v_l))$.
Clearly, the central quasi-distribution is a convex combination of basis quasi-distributions,  $c=\sum_{j=1}^nc_jf_j$. Finally, since the projections of $v_j$'s onto $V_0$ play the crucial role in the characterisation of morphophoric measurements (Theorem \ref{morphTF}) we shall also consider their images by the measurement map $\pi$, i.e. the vectors in $L$ defined by
$\phi_j:=\pi(P_0(v_j))=\pi(e(v_j)(w_j-m))=e(v_j)(f_j-c)$ for $j=1,\dots,n$.

Obviously, $v_j, w_j \in V$, as well as $f_j$ and $c \in \mathcal{A}$, depend on the choice of $m$ and $\mu$, but we omit the subscripts for greater readability. On the other hand, $\mathcal A$, $L$ and $\phi_j \in L$ (see the proof of Theorem \ref{morphTF}) depend only on the measurement $\pi$.  We call $f_j$, $j=1,\ldots,n$, and $c$  \textsl{quasi-}distributions because in general, they do not need to lie inside the probability simplex, i.e. their coordinates sum up to $1$, but are not necessarily non-negative. However, if we assume that $C$ is supra-dual (which can be done by an appropriate choice of $m$ and $\mu$, see Remark \ref{makesupra}), then from $\pi_j \in C^*$ we get $v_j = T_{m,\mu}^{-1}(\pi_j) \in C$. Hence $w_j \in B$, and so $f_j\in \pi(B)\subset\Delta_n\cap\mathcal A$. In this case also $c \in \mathcal{P}$. 

The next theorem is crucial for understanding the geometry of generalised qplexes.
\begin{Th}
\label{morfra}
Let $\pi$ be informationally complete. Then	$\pi$ is morphophoric with the morphophoricity constant $\alpha$  if and only if $(\phi_j)_{j=1}^n$ is a tight frame for $L$ with the frame bound $\alpha^2$.
\end{Th}

\begin{proof}
	Let $\pi$ be morphophoric with the morphophoricity constant $\alpha$ and let  $f\in L$. Then $f = \pi(x)$ for some $x\in V_0$. From Theorem \ref{morphTF} we obtain
	\begin{align*}
	\sum_{j=1}^n\langle f,\phi_j\rangle\phi_j&=\sum_{j=1}^n\langle\pi(x),\pi(P_0(v_j))\rangle\pi(P_0(v_j))\\
	&= \pi(\alpha\sum_{j=1}^n\langle x,P_0(v_j)\rangle_{m,\mu} P_0(v_j))\\
	&=
	\pi(\alpha^2x)=\alpha^2\pi(x)=\alpha^2 f,
	\end{align*}
	as desired.
	
	On the other hand, let us assume now that $(\phi_j)_{j=1}^n$ is a tight frame for $L$ with the frame bound $\alpha^2$. Let $x,y\in V_0$. Denote by $S$ the frame operator  for $(P_0(v_j))_{j=1}^n$. Then from Theorem \ref{tf} we get
	\begin{align*}
		\langle x, Sy\rangle_{m,\mu}&=
		\sum_{j=1}^n \langle x,P_0(v_j)\rangle_{m,\mu}\langle P_0(v_j),y\rangle_{m,\mu} = \sum_{j=1}^n\langle x, v_j\rangle_{m,\mu}\langle v_j,y\rangle_{m,\mu}\\
	&= \sum_{j=1}^n\pi_j(x)\pi_j(y)=\langle \pi(x),\pi(y)\rangle=\frac{1}{\alpha^2}\sum_{j=1}^n\langle \pi(x),\phi_j\rangle\langle \phi_j,\pi(y)\rangle\\
	&=	\frac{1}{\alpha^2}\sum_{j=1}^n\langle \pi(x),\pi(P_0(v_j))\rangle\langle \pi(P_0(v_j)),\pi(y)\rangle\\
	&= \frac{1}{\alpha^2}\sum_{j,k,l=1}^n \pi_k(x)\pi_k(P_0(v_j))\pi_l(P_0(v_j))\pi_l(y)\\
	&= \frac{1}{\alpha^2}\sum_{j,k,l=1}^n \langle x,v_k\rangle_{m,\mu}\langle v_k,P_0(v_j)\rangle_{m,\mu}\langle P_0(v_j),v_l\rangle_{m,\mu}\langle v_l,y\rangle_{m,\mu}\\
	&= \frac{1}{\alpha^2}\sum_{j,k,l=1}^n \langle x,P_0(v_k)\rangle_{m,\mu}\langle P_0(v_k),P_0(v_j)\rangle_{m,\mu}\langle P_0(v_j),P_0(v_l)\rangle_{m,\mu}\langle P_0(v_l),y\rangle_{m,\mu}\\
	&=\frac{1}{\alpha^2}\langle x, S^3y\rangle_{m,\mu}.
	\end{align*}
Thus $S^3=\alpha^2 S$. The informational completeness implies that $S$ is full-rank, see Theorem \ref{IC}. Since the frame operator is positive-semidefinite, $S=\alpha I$. 
\end{proof}

\begin{R}
The set $(\phi_j)_{j=1}^n$ is balanced, i.e. $\sum_{j=1}^n \phi_j = 0$.
\end{R}

From now on we assume that the measurement $\pi$ is morphophoric. First, we observe that a tight frame $(\phi_j)_{j=1}^n$ is a scaled orthogonal projection of the canonical basis.

\begin{Prop}
\label{proj}
Let $e_1,\ldots,e_n$ be the canonical basis of $\mathbb R^n$ and $P:\mathbb R^n\to L$ be the orthogonal projection onto $L$. Then $\phi_j=\alpha Pe_j$ for $j = 1,\ldots,n$.
\end{Prop}
\begin{proof} Let $j,k = 1,\ldots,n$. Then
	\begin{align*}
		\langle \phi_j,\phi_k\rangle&=\langle \pi(P_0(v_j)),\pi(P_0(v_k))\rangle=\alpha\langle P_0(v_j),P_0(v_k)\rangle_{m,\mu}=\alpha\langle v_j,P_0(v_k)\rangle_{m,\mu}\\
		&=\alpha\pi_j(P_0(v_k))=\alpha(\phi_k)_j=\alpha\langle e_j,\phi_k\rangle=\langle \alpha Pe_j,\phi_k\rangle.		
	\end{align*}
	Thus  $\phi_j-\alpha Pe_j$ is orthogonal to $\phi_k$. But vectors $(\phi_k)_{k=1}^n$ span $L$, therefore $\phi_j-\alpha Pe_j=0$ for every $j=1,\dots,n$, as required.
	\end{proof}

\begin{R}
From  Proposition \ref{proj} we get the following symmetry relation: $(\phi_j)_k=(\phi_k)_j$ for $j,k=1,\dots,n$.
\end{R}

Let us now introduce two special polytopes: the \emph{primal polytope} $\Delta:=\mathcal A\cap\Delta_n$ being a $(\dim V_0)-$ dimensional section of the probability simplex $\Delta_n$ and the \emph{basis polytope} $D:=\conv\{f_1,\ldots,f_n\}$ being the convex hull of basis quasi-distributions.
They are related  in the same way as in the quantum theory.

\begin{Th}[dual polytopes]
\label{duapol}
	The polytopes $D$ and $\Delta$ are dual in $\mathcal A$ with respect to the sphere with centre at $c$ and 
	of radius $\sqrt{\alpha\mu}$.
\end{Th}
\begin{proof}
It suffices to show that $p\in \Delta$ if and only if $p\in\mathcal A$ and $\langle p-c,f_j-c\rangle\geq-\alpha\mu$ for every $j=1,\ldots,n$. From Proposition \ref{proj} for $p\in\mathcal A$ we obtain 
	$$\langle p-c,f_j-c\rangle=\frac{1}{e(v_j)}\langle p-c,\phi_j\rangle=\frac{1}{e(v_j)}\langle p-c,\alpha Pe_j\rangle=\frac{\alpha}{e(v_j)}\langle p-c,e_j\rangle=\frac{\alpha}{e(v_j)}p_j-\alpha\mu.$$
Thus $\langle p-c,f_j-c\rangle\geq-\alpha\mu$ for all $j$ if and only if $p_j\geq 0$ for all $j$, which is equivalent to $p\in\Delta$.
\end{proof}

The constant $\alpha\mu$ that appear in the above result is another `dimensionless' quantity. This one is related to the structure of the morphophoric measurement $\pi$, and this is why we call it \textit{measurement constant}. 

\begin{Prop}[sandwich]
\label{sandwich}
The central distribution $c=\sum_{j=1}^nc_jf_j\in D$. Moreover, $\mathcal P\subset\Delta$ and under the assumption of supra-duality of $C$ we have the following inclusions 
\begin{equation}
D\subset \mathcal P\subset\Delta.
\end{equation}
\end{Prop}

For self-dual state space $\mathcal P$ is a self-dual set sandwiched between two dual polytopes: basis and primal.

\begin{Prop}[self-dual generalised qplex]
\label{sdq}If $B$ (and so the corresponding GGPT) is self-dual then $\mathcal P$ is self-dual with respect to the sphere with centre at $c \in D$ and of radius $\sqrt{\alpha\mu}$.	
\end{Prop}
\begin{proof}
	It follows from the self-duality of $B$ and the morphophoricity of $\pi$.
\end{proof}

\subsection{Regular measurements in self-dual spaces}
\label{Regular measurements}
In this section we assume that a GGPT is self-dual. We distinguish a special class of measurements with more regular behaviour leading to more clear geometric structure.

\begin{Df}
\label{regular}
We call the morphophoric measurement $\pi$ \textsl{regular} if 
$\pi$ is \textsl{unbiased}, i.e. $\pi(m) = c = c_n := (1/n,\dots,1/n)$, and the effects $\pi_j$, $j=1,\ldots,n$, lie on the rays of the dual cone maximally distant from the central ray $\{te:t\geq 0\}$, and so necessarily extreme, see Proposition \ref{extreme points}.

\end{Df}
The conditions in the definition of regular measurement can be expressed in terms of $v_j$ and $w_j$ as follows: $e(v_{j})=1/(n\mu)$ and $\left\| w_j - m \right\|_0^2= \chi$, in particular, $w_j$ is a pure state, for $j=1,\dots,n$.
Morphophoric regular measurements are natural counterpart of the rank-$1$
equal norm POVMs generated by so-called $2$-designs in
$\mathcal{P}(\mathbb{C}^{n})$ \cite[Corollary 9]{SloSzy20}, including SIC-POVMs used in the canonical version of QBism. 
For such measurements we have the
following result which binds together four `dimensionless' constants of the theory: the \textsl{measurement constant} $\alpha\mu$, the \textsl{space constant} $\chi/\mu$, the  \textsl{measurement dimension} $n$, and the  \textsl{space dimension} $\dim{V}_{0}$. This theorem generalises the formula for the similarity ratio in \cite[Corollary 9]{SloSzy20}.

\begin{Th}[constants]
\label{formula}
Let $\pi$ be a morphophoric regular measurement in a self-dual GGPT. Then
\begin{equation}
\label{formula2}
\alpha\mu=\frac{\chi/\mu}{n\dim{V}_{0}}\text{.}
\end{equation}
\end{Th}

\begin{proof}
Applying Theorem \ref{morphTF}.(i)$\Rightarrow$(ii) we get
\begin{align*}
\alpha\mu & =\frac{\mu}{\dim{V}_{0}}\left(\sum_{j=1}^{n}\left\|
v_{j}\right\|^2  _{m,\mu}-\mu (e(v_{j}))^{2}\right)\\
& =\frac{\mu}{\dim{V}_{0}}\left(\sum_{j=1}^{n}(e(v_{j}))^{2}(\left\|
w_{j}\right\|^2  _{m,\mu}-\mu)\right)\\
& =\frac{\chi\mu}{\dim{V}_{0}}\sum_{j=1}^{n}(e(v_{j}))^{2}=\frac{\chi/\mu
}{n\dim{V}_{0}}\text{,}
\end{align*}
as desired.
\end{proof}

\begin{R}
Note that for a morphophoric regular measurement in a self-dual GGPT the product of the measurement constant and dimension, $\alpha\mu n = (\chi/\mu)/ \dim V_0$, is given by (the numbering of examples as in Sect. \ref{exGGPT}): (A) $1$ (the classical space), (B) $1/(d+1)$ (the quantum space),  (C) $1/(2\cos(\pi/N))$ (the  $N$-gonal space with $N$ odd), including  $(\sqrt{5}-1)/2$ for $N=5$ (the pentagonal space), and (D) $1/N$ (the $N$-dimensional ball). 
\end{R}

\begin{R}
If $\pi$ is a morphophoric regular measurement in a self-dual GGPT, then the section of $\Delta_n$ by $\mathcal{A}$ is \textsl{central}, i.e $(1/n,\dots,1/n) \in \mathcal{A}$, and \textsl{medial}, i.e the vertices of $\Delta_n$ are equidistant from $\mathcal{A}$. Namely, as in the proof \cite[Theorem 14]{SloSzy20}, one can show, using \eqref{formula2}, that $\operatorname*{dist} (e_j,\mathcal{A}) = \sqrt{1-\frac{\dim {V}}{n}}$, see Fig. 4 from \cite{SloSzy20}.
\end{R}

As a consequence of Proposition \ref{Funine} and Theorem \ref{formula} we get the following bounds for the inner product of two probability vectors in the generalised qplex.

\begin{Prop}[fundamental inequalities for probabilities]
\label{Funinepro}
Let $\pi$ be a morphophoric regular measurement in a self-dual GGPT.
Then, the following inequalities:
\begin{equation}
-\frac{\chi/\mu}{\dim{V}_{0}} \leq n\langle p, q \rangle -1 \leq \frac{(\chi/\mu)^2}{\dim{V}_{0}}
\end{equation}
hold for all $p,q \in \mathcal{P}$, see \cite[Theorem 16]{SloSzy20} for quantum case. The first inequality becomes an equality if and only if $p$ and $q$ are antipodal points in $\mathcal{P}$, whereas the second if and only if $p=q$ is an extremal element of $\mathcal{P}$.
\end{Prop}

The next result follows directly from the previous considerations.

\begin{Prop}[spheres and polytopes]
\label{spheres and polytopes}
If $\pi$ is a morphophoric regular measurement in a self-dual GGPT, then $c=(1/n,\dots,1/n)$ and the sets $B_{\mathcal{A}}(c,\sqrt{\alpha\mu/(\chi/\mu)}\cup D\subset \mathcal{P} \subset B_{\mathcal{A}}(c,\sqrt{\alpha\chi}) \cap \Delta$ are located in such a way that  the outer ball $B_{\mathcal{A}}(c,\sqrt{\alpha\chi})$ is circumscribed about $D \subset \textnormal{ext}(\mathcal{P})$, and, from duality, the inner ball $B_{\mathcal{A}}(c,\sqrt{\alpha\mu/(\chi/\mu)}= B_{\mathcal{A}}(c,1/(n\dim{V}_{0}))$ is inscribed in $\Delta$, see Fig. 5 from \cite{SloSzy20}. Moreover, if $n$ is minimal, i.e. $n=\dim V$, then $D$ is a regular simplex.
\end{Prop}

\section{Morphophoric measurements - examples}
\label{Morphophoric measurements - examples}

Now, let us take a closer look at several examples of morphophoric measurements.\footnote{\textsl{Longum iter est per praecepta, breve et efficax per exempla}, Seneca the Younger, \textsl{Epistulae morales ad Lucilium}.} 
As we already mentioned for the quantum  case all complex projective 2-designs generate the rank-1 equal trace POVM, which are necessarily regular morphophoric, including SIC-POVMs, and in fact, all rank-1 equal-trace morphophoric POVMs can be obtained in this way, see \cite[Corollary 9]{SloSzy20}. However, other GGPTs also provide, as we know from Theorem \ref{existence} and Remark \ref{existence2}, a great number of such measurements. We shall see in this chapter several  examples of objects of this kind.

We start with two polygonal GGPTs, see Example \ref{poly}: square, which is not self-dual, and pentagonal, which is the simplest non-classical case of a two-dimensional self-dual GGPT. The third example lives in the ball GGPT, see Example \ref{exball}.

\begin{Ex}[Square GGPT]\label{squareggpt} Let us set the parameter $a=\sqrt{2}$, i.e. the optimal one guaranteeing the supra-duality, see Example \ref{poly}. 
\begin{enumerate}[a.]
    \item First, we consider the  boundary measurement given by $\pi_j(z_j)=\pi_j(z_{j+1})=\frac{1}{2}$ and $\pi_j(z_{j+2})=\pi_j(z_{j+3})=0$, for $j=1,\ldots,4$ (the addition is mod 4). In other words, $\pi_j=v_j=(\pm\frac{\sqrt{2}}{8},\pm\frac{\sqrt{2}}{8},\frac{1}{4})$ and $w_j=4v_j$. Then $(P_0(v_j))_{j=1}^4$ is a tight frame in $V_0$ with the frame bound $\alpha=\frac{1}{2}$. Note also that $C^+=\{t_1v_1+\ldots+t_4v_4:t_1,\ldots,t_4\geq 0\}$.
    
    The basis distributions are of the form $f_j(j)=\pi_j(w_j)=\frac{1}{2}$, $f_{j+1}(j)=\pi_{j}(w_{j+1})=\frac{1}{4}=f_{j-1}(j)=\pi_{j}(w_{j-1})$ and $f_{j+2}(j)=\pi_j(w_{j+2})=0$. In particular, $\mathcal P=\Delta$ and $D$ is the dual square inscribed into $\Delta$, see Fig.\ref{square}.i.
    \item Next, let us consider a measurement defined by $\pi_j(z_j)=\frac{1}{2}$, $\pi_j(z_{j+1})=\pi_j(z_{j+3})=\frac{1}{4}$ and $\pi_j(z_{j+2})=0$. In other words, $\pi_j=v_j=\frac{1}{8}(z_j+m)$, $w_j=\frac{1}{2}(z_j+m)$ and $(P_0(v_j))_{j=1}^4$ is a tight frame in $V_0$ with the frame bound $\alpha=\frac{1}{4}$. Obviously, $\pi$ is  a boundary measurement.
    
    The basis distributions are of the form $f_j(j)=\pi_j(w_j)=\frac{3}{8}$, $f_{j+1}(j)=\pi_j(w_{j+1})=\frac{1}{4}=f_{j+3}(j)=\pi_j(w_{j+3})$ and $f_{j+2}(j)=\pi_j(w_{j+2})=\frac{1}{8}$. In particular, $\Delta$ is the same as in the previous case, $\mathcal P$ now takes the place of $D$, and new $D$ is the dual square to $\Delta$, see Fig.\ref{square}.ii.
\end{enumerate}
\begin{figure}[htb]
		\includegraphics[scale=0.28]{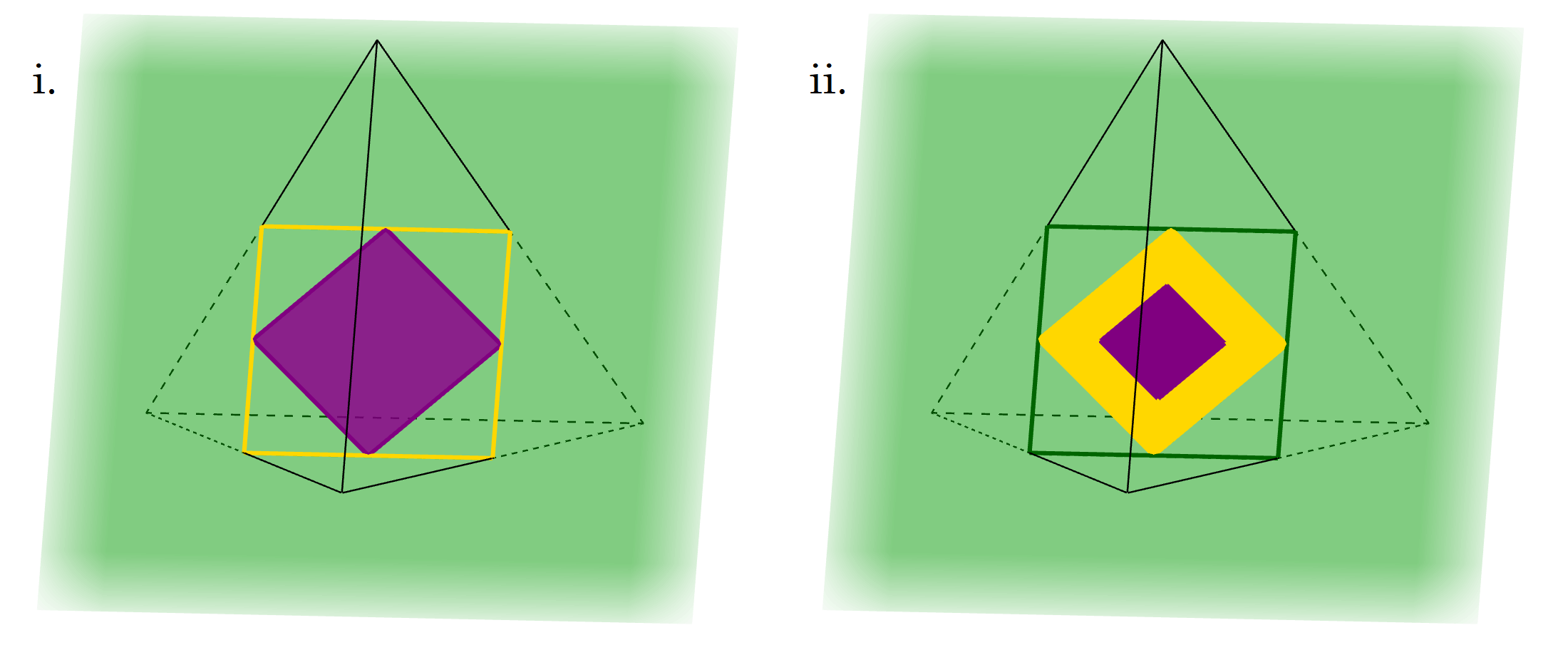}
		\centering
		\caption{The intersection of the simplex $\Delta_4$ by the affine plane $\mathcal A$ (green) with the primal polytope $\Delta$ (square with yellow edges on the left and with green edges on the right),  the set of possible probabilities $\mathcal P$ (the same as $\Delta$ on the left and yellow square on the right) and the dual polytope $D$ (purple inner square) in the scenarios a. and b. of Example \ref{squareggpt} respectively.}
		\label{square}
	\end{figure}
\end{Ex}

\begin{Ex}[Pentagonal state space]\label{penta} Let us set the parameter $a=\sqrt{\sqrt{5}-1}$, i.e. the  one guaranteeing the self-duality, see Example \ref{poly}.
	
	\begin{enumerate}[a.]
		\item First, we consider the most natural choice of measurement, i.e. the pentagonal (regular) one. Put $\pi_j=v_j=\frac{1}{5}z_j$, $j=1,\ldots,5$. Then $(P_0(v_j))_{j=1}^5$ is a tight frame in $V_0$  with the frame bound $\alpha=\frac{\sqrt{5}-1}{10}$. 
		
		The basis distributions are of the form $f_j(k)=\pi_k(z_j)=\frac{1}{5}((\sqrt{5}-1)\cos\frac{2 \pi(k-j)}{5}+1)$, i.e. the vectors $f_j$ are all cyclic permutations of $(\frac{\sqrt{5}}{5},\frac{5-\sqrt{5}}{10},0,0,\frac{5-\sqrt{5}}{10})$. In particular, in this situation $D=\Delta=\mathcal P$.
		\item Let us now consider a minimal morphophoric measurement, i.e. consisting of 3 effects. Put  $\pi_j=v_j=\frac{1}{3}(ra\cos\frac{2 \pi j}{3},ra\sin\frac{2\pi j}{3},1)$, where $j\in\{1,2,3\}$ and $r>0$. Then $(P_0(v_j))_{j=1}^3$ is an equal-norm tight frame in $V_0$  with the frame bound $\alpha=\frac{1}{6}r^2a^2$. Such measurement cannot be boundary  but we can make  $\pi_1$ and $\pi_2$ boundary under appropriate choice of $r$, i.e. taking it maximally possible (the exact formula is quite long and since it is not crucial we decided not to include it here).  
		
		The basis distributions are of the form $f_j(k)=\pi_k(w_j)=\pi_k(3v_j)=\frac{1}{3}(r^2a^2\cos\frac{2\pi(k-j)}{3}+1)$, i.e. the vectors $f_j$ are all cyclic permutations of $\frac{1}{3}(1+r^2a^2,1-\frac{r^2a^2}{2},1-\frac{r^2a^2}{2})$. In other words, $D$ is the image of $\Delta_3$ after a homothety with center $c$ and ratio $\frac{r^2a^2}{2}$ (Fig.\ref{5-3}.i).
	
	\item We can also rotate and rescale the measurement above in the following way: put $\pi_j=v_j=\frac{1}{3}(ra\cos\frac{2 \pi j+\pi}{3},ra\sin\frac{2\pi j+\pi}{3},1)$ with $r$ also maximal possible (this time easy to calculate, $r=-\cos\frac{4\pi}{5}$), so that $\pi_1$ is boundary. 
	
	The basis distributions are of the same form as above (with new $r$), but the polytope $\mathcal P$ is now rotated (Fig.\ref{5-3}.ii).
	
\item Finally, we include some white noise in the case b., i.e.  $\tilde{\pi}_j=\lambda\pi_j+(1-\lambda)q_j e$, where $q_j\geq 0$, $\sum_{j=1}^3q_j=1$ and $\lambda\in(0,1)$. An example is presented on  Fig.\ref{5-3}.iii.

	\end{enumerate}
\begin{figure}[htb]
    \centering
    \includegraphics[scale=0.4]{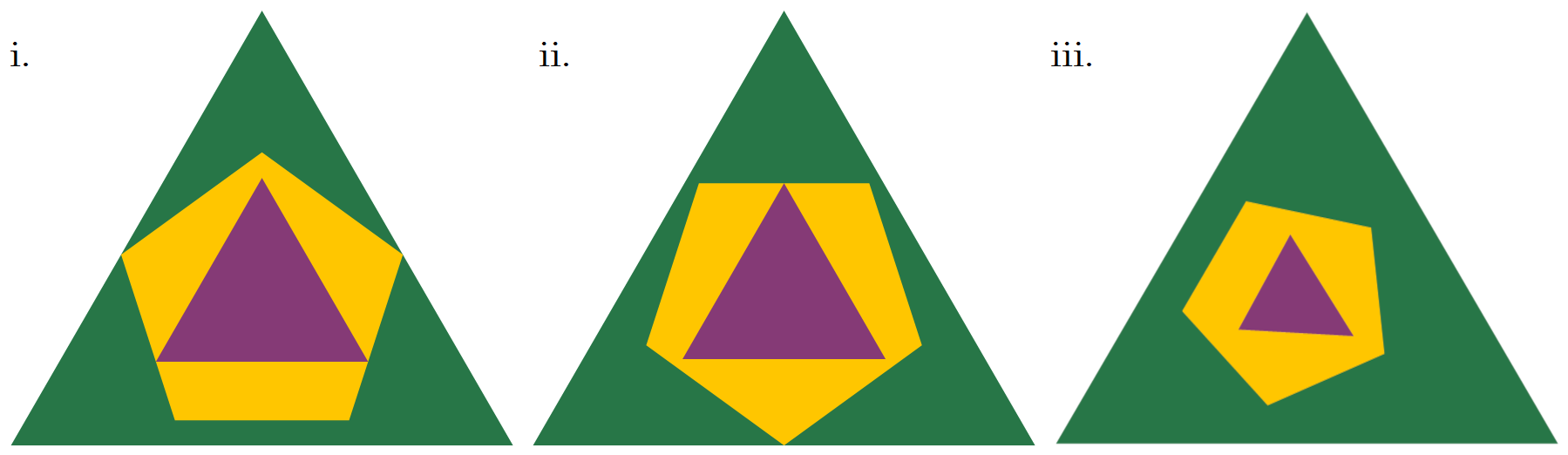}
    \caption{The primal polytope $\Delta$ (green outer triangle), the set of possible probabilities $\mathcal P$ (yellow pentagon) and the primal polytope $D$ (purple inner triangle) in the scenarios b., c. and d. of Example \ref{penta}.}
    \label{5-3}
\end{figure}	
\end{Ex}
\begin{Ex}[Ball state space]\hfill

\begin{enumerate}[a.]
	\item Let us consider a minimal morphophoric measurement, i.e. consisting of $n=N+1$ effects. In contrary to the pentagonal case, this time we can make an equal-norm measurement not only boundary but also regular by setting
	$w_j$, $j=1,\ldots, N+1$ to be the vertices of a regular $n$-dimensional simplex inscribed in $B$ and $\pi_j=v_j=\frac{1}{N+1}w_j$.
	
	The basis distributions satisfy the conditions $f_j(k)=\frac{N-1}{N(N+1)}$, $j\neq k$ and $f_j(j)=\frac{2}{N+1}$.
	In particular, $\Delta=\Delta_{N+1}$ and $\mathcal P$ is the ball inscribed in $\Delta$ and circumscribed on the dual simplex $D$ (see Fig.\ref{ball}).
	\begin{figure}[htb]
		\includegraphics[scale=0.4]{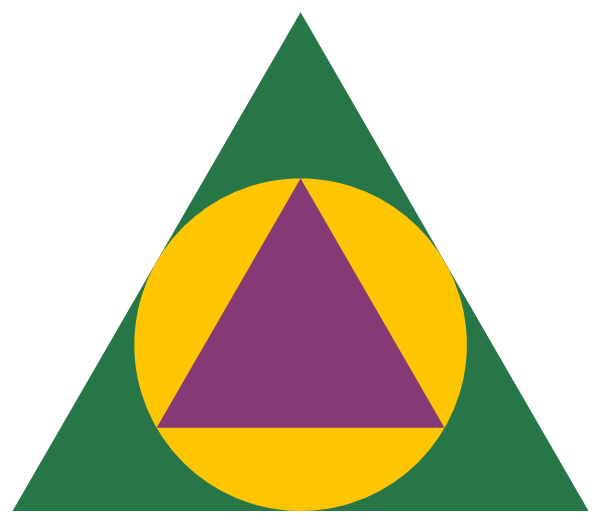}
		\centering
		\caption{The primal polytope $\Delta=\Delta_3$ (green outer triangle), the set of possible probabilities $\mathcal P$ (yellow disk) and the primal polytope $D$ (purple inner triangle) for the minimal equal-norm measurement on the 2-dimensional ball state space.}
		\label{ball}
	\end{figure}

	\item Let us now consider a morphophoric measurement defined by the vertices of the hypercube, $n=2^N$. Again, we want it to be regular and boundary thus $w_j$ are of the form $\frac{1}{\sqrt{N}}(\pm 1,\ldots,\pm 1,\sqrt{N})$ and $\pi_j=v_j=\frac{1}{2^{N}}w_j$, $j=1,\dots,2^N$.
	
In particular, $\mathcal P$ is the ball inscribed in the cross-polytope $\Delta$ and circumscribed on the hypercube $D$ (Fig.\ref{cuboct}.i).
	
	\item Another example is a regular boundary morphophoric measurement defined by the vertices of the cross-polytope, $n=2N$. This time $w_j:=e_j+e_{N+1}$, $w_{j+N}=-w_j+2m$, for $j=1,\ldots,N$,  and $\pi_j=v_j=\frac{1}{2N}w_j$, for $j=1,\ldots,2N$.
	
	The basis distributions satisfy the following conditions: $f_j(j)=\frac{1}{N}$, $f_j(j+N)=0$ and $f_j(k)=\frac{1}{2N}$ for other values of $k$. The primal polytope $\Delta$ is a hypercube with the vertices of the form $g_j(k)\in\{0,\frac{1}{N}\}$ for $k=1,\ldots, N$ and $g_j(k+N)=\frac{1}{N}-g_j(k)$. In particular, $\mathcal P$ is the ball inscribed in the hypercube $\Delta$ and circumscribed on the cross-polytope $D$ (Fig.\ref{cuboct}.ii).

\end{enumerate}
\begin{figure}[htb]
	\includegraphics[scale=0.6]{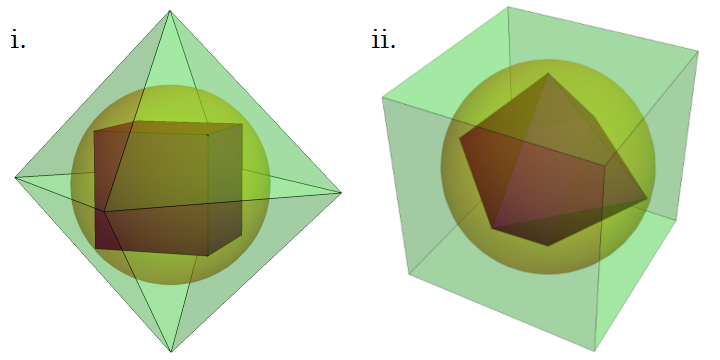}
	\centering
	\caption{The primal polytope $\Delta$ (octahedron on the left, and cube on the right), the set of possible probabilities $\mathcal P$ (ball inscribed in $\Delta$) and the primal polytope $D$ (polyhedron dual to $\Delta$, inscribed in $\mathcal P$) for the (hyper)cube (on the left) and the cross-polytope (on the right) measurement on the 3-dimensional ball state space.}
	\label{cuboct}
\end{figure}
\end{Ex}

\section{The primal equation}
\label{The primal equation}

\subsection{Instruments}
\label{Instruments}
Throughout this section we assume that the GGPT in question is supra-dual.
The measurement alone gives us just the probability distributions of the measurement outcomes. If we want to describe what happens to the system after performing the measurement we need to define a measurement instrument. 

\begin{Df}
Let $\pi:=(\pi_{j})_{j=1}^{n}$ be a measurement and let $\Lambda :=(\Lambda_{j})_{j=1}^{n}$ be a family of affine maps from $B$ to $C$.  We say that $\Lambda$ is an \textsl{instrument} for $\pi$ if $\pi_{j}(x)=e(\Lambda_{j}(x))$ for every \textsl{pre-measurement} state $x \in B$ and for all $j=1,\dots,n$. We assume that the \textsl{post-measurement} state is given by  $\Lambda_{j}(x)/e(\Lambda_{j}(x))$, supposing that the result of measurement was actually $j$ (and so $\pi_j(x)\neq 0$) \cite[Ch. 6]{Slo03}. Clearly, 
$\sum_{j=1}^{n} \Lambda_{j}$ is an affine operator from $B$ to $B$. It assigns to a given prior the state of the system after the measurement $\pi$ has been performed while the result of the measurement was unknown. 
\end{Df}

A measurement instrument allows us to describe the probabilities of the outcomes of subsequent measurements. Namely, let us denote by $p_{jk}^{\pi\xi}(x)$ the \textsl{probability} that the results of subsequent measurements $\pi=(\pi_j)_{j=1}^n$ and $\xi=(\xi_k)_{k=1}^{n'}$ on the initial state $x$ were $j$ and $k$, and by $p_{k|j}^{\xi|\pi}(x)$ the \textsl{conditional probability} that the result of measurement $\xi$ is $k$ given that we measured $j$ with measurement $\pi$ for $j=1,\dots,n$ and $k=1,\dots,n'$. Then $p_{k|j}^{\xi|\pi}(x)=\xi_k(\Lambda_j(x)/e(\Lambda_j(x)))$ and $p_{jk}^{\pi\xi}(x)=p_{k|j}^{\xi|\pi}(x)\cdot\pi_j(x)=\xi_k(\Lambda_j(x))$.

Let us observe that for any instrument $\Lambda$ and $j=1,\dots,n$ we have 
\begin{equation}
e(\Lambda_j(m))=\pi_j(m)=\langle m, T_{m,\mu}^{-1}(\pi_j)\rangle_{m,\mu}=\mu e( T_{m,\mu}^{-1}(\pi_j)) = \mu e(v_j),
\end{equation}
where $v_j:=T_{m,\mu}^{-1}(\pi_j)$.
Thus, thanks to the supra-duality, we can naturally distinguish a special class of instruments, for which 
\begin{equation}
\label{sdinstrument}
\Lambda_j(m)=\mu T_{m,\mu}^{-1}(\pi_j)=\mu v_j
\end{equation}
for $j=1,\dots,n$ or, equivalently, $\Lambda_j(m)/\pi_j(m) =v_j/e(v_j):= w_j$.
It is easy to see that condition \eqref{sdinstrument} can be also written as 
\begin{equation}
\label{instrument}
\langle \Lambda_j(m),x\rangle_{m,\mu}=\langle m,\Lambda_j(x)\rangle_{m,\mu} 
\end{equation}
for $x \in B$. Note that if \eqref{instrument} holds, then it is also true for every $x \in V$, where we consider the natural extension of $\Lambda_j$ to the full space $V$. Thus \eqref{instrument} (and \eqref{sdinstrument}) are equivalent to
\begin{equation}
\label{instrument2}
\Lambda_j(m)=\Lambda_j^*(m),
\end{equation}
for  $j=1,\dots,n$, where the dual maps $\Lambda_j^*$ are taken with respect to the inner product $\langle \cdot,\cdot\rangle_{m,\mu}$. We call such instrument $\Lambda$ \textsl{balanced at} $m$. Note that $\sum_{j=1}^{n} \Lambda_{j}(m)=m$ holds in this case, i.e. if the state of the system \textsl{before} the measurement $\pi$ is $m$ and the result of the measurement is unknown then the state of the system \textsl{after} the measurement remains unchanged.

Despite many equivalent formulations above, the meaning of our key assumption of balancing at $m$ may at this point remain elusive to the reader. 
However, for a self-dual state space we are able to make it a bit clearer.
For an arbitrary instrument $\Lambda$ let us make the necessary assumption that $\sum_{j=1}^{n} \Lambda_{j}(m)=m$. Using the self-duality, we deduce that $\Lambda^*:=(\Lambda^*_j)_{j=1,\dots,n}$ is an instrument for the measurement given by $\accentset{\ast}{\pi}_j(x) := \langle \Lambda_j(m),x \rangle_{m,\mu}/\mu$ for $x \in B$ and $j=1,\dots,n$. It is because for $x,y \in B$ and $j=1,\dots,n$, from the infra-duality, we have $\langle \Lambda_j^*(x),y \rangle_{m,\mu}=\langle x,\Lambda_j(y) \rangle_{m,\mu} \geq 0$, and so, from the supra-duality, we get $\Lambda_j^*(x) \in C$. We call $\Lambda^*$ \textsl{retrodiction instrument}, see \cite[Theorem 6.4]{Slo03} for the classical GGPT case and Example \ref{Lud} and \ref{Ludgen} for quantum GGPT case. Obviously,
\begin{equation}
\label{instrument3}
\pi_j(m) =\accentset{\ast}{\pi}_j(m)
\end{equation}
for $j=1,\dots,n$. Moreover, observe that in this situation the assumption \eqref{instrument2} is equivalent to
\begin{equation}
\label{instrument3*}
\pi_j(x) =\accentset{\ast}{\pi}_j(x)
\end{equation}
for $x \in B$ and $j=1,\dots,n$, and then $\Lambda^*$ is also an instrument for $\pi$.

 We now turn to the connections between retrodiction and the \textsl{Bayes formula}. Let $j,k=1,\dots,n$. Consider the probabilities and the conditional probabilities of the results of two subsequent measurements. In the first case, we start from $\accentset{\ast}{\pi}$ with the measurement instrument $\Lambda^*$, and then we use $\accentset{\ast}{\pi}$ again (expressions with star). In the second case, we do the same for $\pi$ with the measurement instrument $\Lambda$, applying again $\pi$ in the second step (expressions without stars). Writing the equalities below, we omit the superscripts denoting the measurements $\accentset{\ast}{\pi}$ and $\pi$, respectively, obtaining: $p^*_{kj}(m)=p_{j|k}^*(m)\accentset{\ast}{\pi}_k(m)=\accentset{\ast}{\pi}_j(\Lambda_k^*(m))=\langle \Lambda_j(m),\Lambda_k^*(m) \rangle_{m,\mu}/\mu
=\langle \Lambda_k(\Lambda_j(m)),m \rangle_{m,\mu}/\mu=
\pi_k(\Lambda_j(m))=p_{k|j}(m)\pi_j(m)=p_{jk}(m)$. Hence, using \eqref{instrument3}, we get
\begin{equation}
\label{instrument4}
p_{j|k}^*(m)=\frac{p_{k|j}(m)\pi_j(m)}{\pi_k(m)}
\end{equation}
without any additional assumptions on $\Lambda$.

Assume now that $\Lambda$ is balanced at $m$. In this situation \eqref{instrument2} implies $\mu p^*_{kj}(m)=\langle \Lambda_j(m),\Lambda_k^*(m) \rangle_{m,\mu}=\langle \Lambda_j^*(m),\Lambda_k(m) \rangle_{m,\mu}=
\langle m,\Lambda_j(\Lambda_k(m)) \rangle_{m,\mu}=\mu p_{kj}(m)$ and, in consequence, $p_{j|k}^*(m)=p_{j|k}(m)$. Now, we obtain from \eqref{instrument4}
\begin{equation}
\label{instrument5}
p_{j|k}(m)=\frac{p_{k|j}(m)\pi_j(m)}{\pi_k(m)}=\frac{p_{k|j}(m)\pi_j(m)}{\sum_{j=1}^n p_{k|j}(m)\pi_j(m)},
\end{equation}
which is the GGPT counterpart of the \textsl{classical Bayes formula} `at the equilibrium point $m$' deduced from \eqref{instrument2}, again see \cite[Theorem 2.4]{Slo03} for the classical GGPT case. In this Bayesian behaviour of our instrument lies the deep meaning of the assumption of the instrument being balanced at equilibrium $m$.

The canonical instruments defined below provide a special case of this construction.

\begin{Df}
\label{canon}
The \textsl{canonical instrument} can be defined for an arbitrary morphophoric measurement $\pi$ and is given by 
\begin{equation}
	\Lambda_j(x):=\pi_j(x)w_j=\langle v_j,x\rangle_{m,\mu} w_j,
\end{equation}
	where the posterior states $\Lambda_j(x)/\pi_j(x)=w_j=v_j/e(v_j)$ are independent of the choice of an initial state $x \in B$ for every $j=1,\dots,n$. Obviously, this instrument is balanced at $m$. In fact, it is self-dual on $V$. In quantum mechanics it is an example of so called \textsl{conditional state preparator} \cite{HeiZim12} or \textsl{Holevo instrument} \cite{Gud22}. Note that in this case the instrument $\Lambda$ acts on the set $\{w_j:j=1,\dots,n\}$ as the (reversible) Markov chain with the probabilities $(\langle v_j,v_k\rangle/e(v_j))_{j,k=1,\dots,n}$ and the initial vector $(\mu e(v_j))_{j=1,\dots,n}$.
\end{Df}

The following result is straightforward.

\begin{Prop}
\label{IE}
Let $\pi$ be a morphophoric measurement in a supra-dual GGPT, and let $\Lambda$ be an instrument for $\pi$. Then the following two conditions are equivalent:
\begin{enumerate}[i.]
\item 
$\Lambda$ is canonical,
\item
\begin{enumerate}[a.]

\item
$\Lambda$ is balanced at $m$ and
\item
the posterior states for $\Lambda$ are independent of priors, i.e. $\Lambda_{j}(x)/p_j(x)=\operatorname*{const}(x\in B, p_j(x) \neq 0)$ for $j=1,\dots,n$.
\end{enumerate}
\end{enumerate}
Moreover, in this situation $\Lambda_j$ is self-dual with respect to $\langle \cdot,\cdot\rangle_{m,\mu}$ for each $j=1,\dots,n$.
\end{Prop}

\begin{Ex}
\label{Lud}
A standard example of a quantum instrument balanced at $m=I/d$ is given by the \emph{generalised L\"uders instrument} $\Lambda_j(\rho):=\Pi_j^{1/2}\rho\Pi_j^{1/2}$ for a density operator $\rho$ and $j=1,\dots,n$, where $\{\Pi_j\}_{j=1}^n$ is a POVM in $\mathbb{C}^d$. This example provides a canonical (\emph{von-Neumann-L\"uders}) instrument when $\{\Pi_j\}_{j=1}^n$ is a PVM (\emph{projection valued measure}) and the projections $\Pi_j$ are one-dimensional.
In this case $\Lambda_j(\rho)= \operatorname*{tr} ({\Pi_j \rho)}\Pi_j$ and $\pi_j(\rho)= \operatorname*{tr} ({\Pi_j \rho)}$. However, if the projections in the PVM are not rank-1, then the generalised L\"uders instrument is not canonical.
\end{Ex}

\begin{Ex}
\label{Ludgen}
Let us consider now a more general instrument for a POVM  $\{\Pi_j\}_{j=1}^n$ in $\mathbb{C}^d$, given by 
$\Lambda_j(\rho):=A_j\rho A^*_j$ for a density operator $\rho$, where $A_j$ are any operators fulfilling $A_j^* A_j = \Pi_j$ ($j=1,\dots,n$). Note that in this case if the measurement is morphophoric, and so in particular informationally complete, then $n \geq d^2$.  Additionally, we assume that $\Lambda$ is \textsl{unital}, i.e., $\sum_{i=1}^n A_i A^*_i = I$, which is equivalent to $\sum_{j=1}^n\Lambda_j(m)=m$. In this case $\Lambda_j^*(m)=\Pi_j/d$ and $\accentset{\ast}{\pi}_j(m)=\tr(\Pi_j)/d$. Moreover, if the state of the system before this measurement $\accentset{\ast}{\pi}$ is $m$ and the result of the measurement is $j$, then the state of the system after this measurement is $\Lambda_j^*(m)/\accentset{\ast}{\pi}_j(m) = \Pi_j/\tr(\Pi_j)$, called by some authors \textsl{retrodictive state} \cite{Baretal00,Jefetal02}. Note that in this case $\Lambda$ is balanced at $m$ if and only if $A_i$ are \textsl{normal}. 
\end{Ex}

\subsection{Various facets of Urgleichung}
\label{Various facets of Urgleichung}
The next theorem states that the morphophoric measurements can be characterised by the generalised primal equation. It allows us to express the probabilities of the outcomes of an arbitrary measurement $\xi$ at the state $x$ in terms of the probabilities of the results of the morphophoric measurement $\pi$ also at $x$, and the probabilities of both measurements $\pi$ and $\xi$ if the initial state is the distinguished state $m$. 

\begin{Th}[Primal equation]
\label{pe}
Let $\pi=(\pi_{j})_{j=1}^{n}$ be a  measurement in a supra-dual GGPT with an instrument
	$(\Lambda_{j})_{j=1}^{n}$ balanced at $m$.  Then the following conditions are equivalent:
	\begin{enumerate}[i.]
		\item $\pi$ is morphophoric.
		\item For any measurement $\xi=(\xi_k)_{k=1}^{n'}$ 
		\begin{equation}
		\label{primeq}
		\delta_{\xi}=\frac{1}{\mu\alpha}\mathsf{C}\delta_{\pi}
		\end{equation}
		holds, where
		$\delta_{\pi}:=\pi\circ P_0$, $\delta_{\xi}:=\xi\circ P_0$, $\mathsf{C}_{kj}:=\xi_k(\Lambda_j(m))-\pi_{j}(m)\xi_{k}(m)$ for $j=1,\ldots,n$ and $k=1,\ldots,n'$, and 
		\begin{equation}
		\label{alpha}
		\alpha:= \frac{1}{\mu\dim \mathcal{A}}\sum_{j=1}^n (\pi_j(\Lambda_j(m))-(\pi_j(m))^2).
		\end{equation}
		\item For some informationally complete measurement $\xi=(\xi_k)_{k=1}^{n'}$ and some $\alpha>0$ \[
		\delta_{\xi}=\frac{1}{\mu\alpha}\mathsf{C}\delta_{\pi}
		\] holds, with $\delta_{\pi},\delta_{\xi}$, and $\mathsf{C}$ are as above.
	\end{enumerate}
	
\end{Th}

\begin{proof} Let $x\in V$, $j=1,\ldots,n$, and $k=1,\ldots,n'$. Recall that $P_0(x)=x-e(x)m$, and so $\delta_{\pi}(x) = \pi(x)-e(x)\pi(m)$, $\delta_{\xi}(x) = \xi(x)-e(x)\xi(m)$. Consequently, we have $\delta_{\pi}(x) = \pi(x)-\pi(m)$ and $\delta_{\xi}(x) = \xi(x)-\xi(m)$ for $x \in B$. We start with the following observations:
		\begin{align*}
				(\delta_{\xi}(x))_k&=\xi_k(P_0(x))=\langle T_{m,\mu}^{-1}(\xi_k),P_0(x)\rangle_{m,\mu}=\langle P_0(T_{m,\mu}^{-1}(\xi_k)),P_0(x)\rangle_{m,\mu},\\
		\frac{1}{\mu}\mathsf{C}_{kj}&=\frac{1}{\mu}(\xi_k(\Lambda_j(m))-\pi_j(m)\xi_k(m))= \langle\xi_k,\pi_j\rangle_{m,\mu}-\mu\langle\pi_j,e\rangle_{m,\mu}\langle\xi_k,e\rangle_{m,\mu} \\ &=\langle\mathcal{P}_0(\xi_k),\mathcal{P}_0(\pi_j)\rangle_{m,\mu}=\langle P_0(T_{m,\mu}^{-1}(\xi_k)),P_0(v_j)\rangle_{m,\mu}.
		\end{align*}
	Thus
	\begin{equation*}	
		\frac{1}{\mu}(\mathsf{C}\delta_\pi(x))_k=\frac{1}{\mu}\sum_{j=1}^n \mathsf{C}_{kj}(\delta_\pi(x))_j=\sum_{j=1}^n\langle P_0(T_{m,\mu}^{-1}(\xi_k)),P_0(v_j)\rangle_{m,\mu}\langle P_0(v_j),P_0(x)\rangle_{m,\mu}.
	\end{equation*}
Note also that \begin{align*}
\frac{1}{\mu}\pi_j(\Lambda_j(m))&=\pi_j(T_{m,\mu}^{-1}(\pi_j))=\langle T_{m,\mu}^{-1}(\pi_j),T_{m,\mu}^{-1}(\pi_j)\rangle_{m,\mu}=\|T_{m,\mu}^{-1}(\pi_j)\|^2_{m,\mu},\\
\frac{1}{\mu}(\pi_j(m))^2&=\frac{1}{\mu}(\langle T_{m,\mu}^{-1}(\pi_j),m\rangle_{m,\mu})^2=\mu (e(T_{m,\mu}^{-1}(\pi_j)))^2.
\end{align*}
	The implication (i.)$\,\Rightarrow\,$(ii.) follows from  Theorem \ref{tf}.\ref{tf2}., (ii.)$\,\Rightarrow\,$(iii.) is obvious, and finally (iii.)$\,\Rightarrow\,$(i.) follows from the informational completeness of $\xi$, see implication (\ref{ICbas}.)$\,\Rightarrow\,$(\ref{ICV0}.) in Theorem \ref{IC}.	
\end{proof}

It turns out that the previous equation can be equivalently expressed in \textsl{purely probabilistic language}. Note that the only constant in this equation is the dimension of the set of states or, in other words, the dimension of the generalised qplex, $\dim B = \dim \mathcal{A}
$.
\begin{Cor}[Primal equation - probabilistic version]
\label{pepr}
Let $\pi=(\pi_{j})_{j=1}^{n}$ be a morphophoric measurement in a supra-dual GGPT with an instrument
	$(\Lambda_{j})_{j=1}^{n}$ balanced at $m$, and let $\xi=(\xi_{k})_{k=1}^{n'}$ be an arbitrary measurement. 
	The equations \eqref{primeq} and \eqref{alpha} above can be written in purely probabilistic terms (we put $p_j^\pi := \pi_j$ and $p_k^\xi := \xi_k$) as
	\begin{equation}
	\label{probab}
		p_k^\xi(x)-p_k^\xi(m)=\dim \mathcal A\cdot \frac{\sum_{j=1}^n\left(p_{jk}^{\pi\xi}(m)-p_j^\pi(m)p_k^\xi(m)\right)\left(p_j^\pi(x)-p_j^\pi(m)\right)}{\sum_{j=1}^n\left(p_{jj}^{\pi\pi}(m)-(p_j^\pi(m))^2\right)}	
	\end{equation}
	for $k=1,\ldots,n'$ and $x\in B$,
	or, equivalently, by using the conditional probabilities
	\begin{equation}
	\label{conditional}
		p_k^\xi(x)-p_k^\xi(m)=\dim \mathcal A\cdot \frac{\sum_{j=1}^np_j^\pi(m)\left(p_{k|j}^{\xi|\pi}(m)-p_k^\xi(m)\right)\left(p_j^\pi(x)-p_j^\pi(m)\right)}{\sum_{j=1}^np_j^\pi(m)\left(p_{j|j}^{\pi|\pi}(m)-p_j^\pi(m)\right)}	
	\end{equation}
for $k=1,\ldots,n'$ and $x\in B$.
\end{Cor}

The above equation is a generalisation of a formula derived in \cite{FucSch09} for SIC-POVMs and known in QBism as the \emph{primal equation} or \emph{Urgleichung}.  In our previous paper \cite{SloSzy20} we derived the generalisation of Urgleichung for arbitrary morphophoric quantum measurement. Now we generalise it even further to cover morphophoric measurements for arbitrary generalised probabilistic theory with supra-dual state space, including also the classical space $\Delta_n$. The primal equation established for SIC-POVMs, takes very simple and elegant form that resembles but at the same time significantly differs from the \emph{law of total probability}. Here the difference is even more visible, but the formula is still expressed in purely probabilistic terms.

Note that in most of QBist literature the independence of the post-measurement states from the choice of initial state is built in from the start, as the authors consider only conditional state preparators as the reference  instruments \cite{Wei24}. Clearly, this is true for SIC-POVM measurements and post-measurement (L\"uders) preparations. However, 
as Khrennikov already noted in his book \cite[p.~182]{Khr16}, even in the quantum case it is interesting to extend this approach to a broader class of measurements and instruments.
We have shown above that, somewhat surprisingly, the assumption of the canonicity of the reference instrument is in no way necessary for the \emph{Urgleichung-like formulae} \eqref{primeq}, \eqref{probab} or \eqref{conditional}
to hold. What is crucial here is only the \emph{morphophoricity of the reference measurement} and the fact that the \emph{associated instrument is balanced at $m$}. But these seem to be very general assumptions regarding rather the nature of the reference measurement process not quantum mechanics as such.

Looking at the Urgleichung-like formula \eqref{conditional} above, we see that in fact six types of probabilities are involved here: 

\begin{itemize}
\item of the \emph{sky (reference)} measurement outcomes: 
\begin{itemize}
\item 
(a) in a given initial state $x$, $p^\pi_j(x)$, 
\item 
(b) in the distinguished initial state $m$, $p^\pi_j(m)$;
\end{itemize}

\item of the \emph{ground} measurement outcomes: 
\begin{itemize}
\item 
(c) in a given initial state $x$, $p^\xi_k(x)$,
\item 
(d) in the distinguished initial state $m$, $p^\xi_k(m)$;
\end{itemize}

\item
(e) the \emph{conditional of the ground} measurement outcomes given the sky measurement results in the distinguished initial state $m$, $p^{\xi|\pi}_{k|j}(m)$; 

\item 
(f) the \emph{conditional of the sky} measurement outcomes given the sky measurement results in the distinguished initial state $m$, $p^{\pi|\pi}_{j|j}(m)$. 

\end{itemize}
Moreover, if the instrument is balanced at $m$, then for the initial state $m$ the probabilities (b), (d), and (e) fulfill the classical \emph{total probability formula}:
\begin{equation}
\label{tpf}
p_{k}^{\xi}\left( m\right)  =
\sum_{j=1}^{n} p_{k}^{\xi}\left( \Lambda_j(m)\right) =
\sum_{j=1}^{n} p_{j}^{\pi}\left(m\right)
p_{k|j}^{\xi|\pi}\left(m\right).
\end{equation}

In fact, the Urgleichung-like formula (see Fig.\ref{diagram}) provides a recipe for obtaining the factual probabilities (c) from the counterfactual probabilities (a) using probabilities (b), (d), (e), and (f) as auxiliary ingredients.

\begin{figure}[htb]
	\includegraphics[scale=0.5]{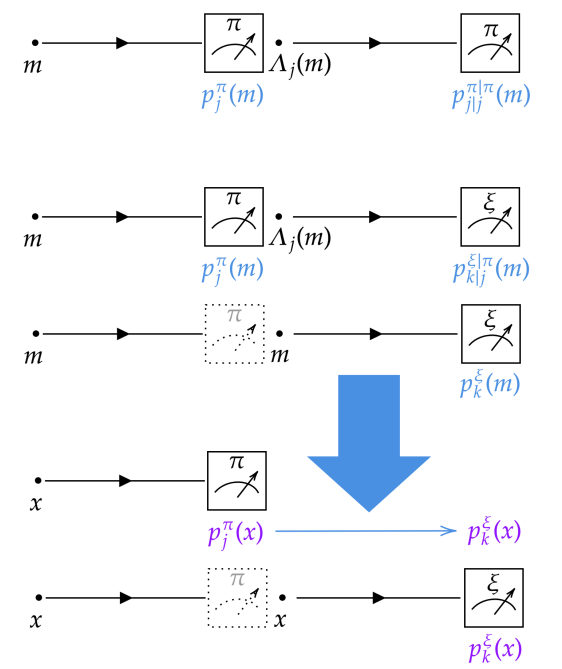}
	\centering
	\caption{In QBism the Urgleichung can be seen as a generalisation of the total probability law and thus a relation between probabilities obtained in two different scenarios: by performing a (morphophoric) measurement $\pi$ first and then applying an arbitrary measurement $\xi$, or by performing $\xi$ directly. However, this interpretation is valid only under the assumption that the measurement instrument for $\pi$ is a conditional state preparator. Under much less restrictive assumption of the instrument being balanced at $m$ the probabilities obtained in these two scenarios for the initial state $m$ actually follow the classical law of total probability. Thus, in some sense, they can be interpreted as \emph{classical} parameters (the blue ones) of the model. In consequence, the Urgleichung-like formula for GGPTs simply connects the probabilities of $\pi$ and $\xi$ outcomes (the purple ones) directly via classical parameters.}
	\label{diagram}
\end{figure}

On the other hand, we shall see that for \emph{canonical instruments} the primal equation takes a particularly simple, though a bit artificial, as it is valid only under these quite restrictive assumptions, \textsl{total probability like} form. To get such equation, i.e., the \textsl{Generalised Urgleichung}, see eqs. (142) from \cite{FucSch09} or (10) from \cite{FucSch11}, it is enough to assume that a morphophoric measurement $\pi$ 
is unbiased, see Definition \ref{regular}, and its instrument is canonical, see Definition \nolinebreak  \ref{canon} and Proposition \ref{IE} for interpretation.

\begin{Prop}[Primal equation - canonical instrument]
\label{pec}
For a canonical instrument of an unbiased morphophoric measurement $\pi$ in a supra-dual GGPT, and for an arbitrary measurement $\xi$ the formula
\begin{equation}
\label{gur}
p_{k}^{\xi}\left( x\right)  =\sum_{j=1}^{n}p_{k|j}^{\xi|\pi}\left(x\right)  (Ap_{j}^{\pi}\left(x\right)  -G),
\end{equation}
holds for $x \in B$ and $k=1,\ldots,n'$, with $A$ and $G$ given by $A:=(n\alpha\mu)^{-1}$ and $G:=n^{-1}(A-1)$.
\end{Prop}

\begin{proof}
It is enough to note that $p_j^{\pi}(m)=1/n$ (from the unbiasedness), $p_{k|j}^{\xi|\pi}(x)=p_{k|j}^{\xi|\pi}(m)$ (from the independence of posterior states on priors), and \eqref{tpf} holds.
Hence and from Corollary \ref{pepr} the assertion follows.
\end{proof}

\begin{R} Alternatively, we can rewrite \eqref{gur} in the following simple form
\begin{equation}
\label{gur2}
p_{k}^{\xi}\left(x\right)  =\sum_{j=1}^{n}p_{k|j}^{\xi|\pi}\left(x\right)  p_{j}^{\pi}\left(x\right)  +(1-1/A)(p_{k}^{\xi}\left(  x\right)-p_{k}^{\xi}\left(  m\right)),
\end{equation}
for $x \in B$, $k=1,\dots,n'$ and $A$  as   above. It is just the \textsl{classical  total probability formula plus a correction term}.
\end{R}

\begin{R}Another form of \eqref{gur}:
\begin{equation}
\label{gur3}
p_{k}^{\xi}\left( x\right) - p_{k}^{\xi}\left( m\right) = A 
\left( \sum_{j=1}^{n} p_{j}^{\pi}\left(x\right)p_{k|j}^{\xi|\pi}\left(x\right)  -p_{k}^{\xi}\left( m\right)
\right),
\end{equation}
was called \emph{Protourgleichung} in \cite[eq. (7.4)]{Wei24}.
\end{R}

\begin{R}
If the measurement is additionally regular in a self-dual space, see Definition \ref{regular}, then it follows from Theorem \ref{formula} that $A=\frac{\dim{V}_{0}}{\chi/\mu}$. If the space is additionally spectral, then $A$ is just the quotient of two dimensions of the state space both reduced by one: the linear and the orthogonal, see Proposition \ref{ort}. Thus $A$ varies from $1$ (for the regular simplex) to $\dim V_0$ (for the ball), see Proposition~\ref{extreme}.
\end{R}

If the original Urgleichung is truly a quantum generalisation of the law of total probability, we should obtain both laws, QBistic and classical, from our general formula for primal equation. Let us show that this is indeed the case. 

We recover exactly the \textsl{classical total
probability formula}, or \textsl{Law of Total Probability} ($A=1$, and so $G=0$) if and only if $\alpha\mu=1/n$. Taking into account
that
\[
\alpha\mu=\frac{1}{\dim \mathcal{A}}\sum_{j=1}^{n}\left(  p_{jj}^{\pi\pi}\left(
m\right)  -(p_{j}^{\pi})^{2}\left(  m\right)  \right)
\]
this condition transforms into
\begin{align*}
\dim V &  =\dim \mathcal{A}+1=n\sum_{j=1}^{n}\left(  p_{j|j}^{\pi|\pi}(m)p_{j}^{\pi}\left(m\right)-(p_{j}^{\pi})^{2}\left(  m\right)  \right)
+1\\
&  =\sum_{j=1}^{n}\left(  p_{j|j}^{\pi|\pi}(m)-p_{j}^{\pi}\left(
m\right)  \right)  +1=\sum_{j=1}^{n}p_{j|j}^{\pi|\pi}(m)\text{.}
\end{align*}
If we assume additionally that $\pi$ is \textsl{repeatable} at $m$, i.e.
$p_{j|j}^{\pi|\pi}(m)=1$ for $j=1,\ldots,n$, then we get in the above
situation $\dim V=n$. An example is provided by the \textsl{classical measurement} $\pi$ in the $N$-dimensional
classical GGPT (see Example (A)) given by $p_{j}^{\pi}(x):=x_{j}$ and $\Lambda_{j}^{\pi}(x):=x_je_{j}$ for $x\in\Delta_{n}$ and $j=1,\dots,n$, where $n=N, \mu=1/N$, and $\alpha =1$.

On the other hand, for the $d$-dimensional \textsl{quantum GGPT}, let us consider an unbiased morphophoric measurement $\pi$ given by a rank-1 equal-trace $n$-element POVM, i.e. a 2-design POVM \cite[Corollary 19]{SloSzy20} along with the corresponding generalised L\"uders canonical instrument, see Example \ref{Lud}. In this case  $\alpha=d/((d+1)n)$ (see \cite[Corollary 9]{SloSzy20}), $\mu=1/d$ (see Example (B)). Then, it is easy to show that $A=d+1$ and $G=d/n$. In this way we obtain eq. (22) from \cite{SloSzy20}. In particular for SIC-POVMs used in QBism we have $G=1/d$, as $n=d^2$, and so we recover the original Quantum Law of Total Probability or, in other words, Urgleichung \cite{FucSch09,FucSch11}. 

\section{Conclusions}

When we first looked at the possible generalisations of the standard QBist approach we did not expect to go so far with so many properties and relations keeping so elegant form. While we do agree that they are especially appealing for the self-dual spaces with the minimal regular measurements (e.g. quantum state space with SIC-POVM), it needs to be emphasised that the core of their beauty lies in the morphophoricity of the reference measurement. Indeed, the key idea that might have escaped  attention   is that since we are interested only in the states, it suffices to keep an eye on what is happening on the linear subspace corresponding to their Bloch representations - and its image by the measurement. And that is what a morphophoric measurement does in the best possible way: keeping the geometry intact and therefore providing simple and elegant reconstruction formulas for states and, in consequence, for the probabilities of arbitrary measurement outcomes, known in standard QBism as primal equations or Urgleichung.

\section*{Acknowledgements}

We express our gratitude to the
anonymous reviewers for their valuable comments and suggestions.
We also thank Tomasz Zastawniak for remarks that improved the readability of the paper. 
AS is supported by Grant No.\ 2016/21/D/ST1/02414 of the Polish National Science Centre. WS is supported by Grant No. 2015/18/A/ST2/00274 of the Polish National Science Centre. Both authors are supported by the Priority Research Area SciMat under the program Excellence Initiative – Research University at the Jagiellonian University in Kraków and acknowledge funding by the European Union under ERC Advanced Grant TAtypic, project number 101142236.

\appendix 
\section{Proof of Theorem \ref{compat} }\label{AppendixProof}

\begin{proof}[Proof of Theorem \ref{compat}]
Note that for every affine $f : V \rightarrow \mathbb{R}$ we have $\min_{B}f=\min_{\partial{B}}f=\min_{\operatorname*{ex}B}f$.
\begin{enumerate}[i.]
\item Cone $C$ is infra-dual with respect to $\left\langle
\cdot,\cdot\right\rangle _{m,\mu}$
if and only if $\left\langle x-e(x)m,y-e(y)m\right\rangle _{0}+e(x)e(y)\mu=\left\langle
x,y\right\rangle _{m,\mu}\geq0$ for $x,y\in C$. Assuming $x,y\neq0$ and
dividing both sides by $e(x)e(y)$ we get $\left\langle x-m,y-m\right\rangle
_{0}\geq-\mu$ for $x,y\in B$. Thus an equivalent condition for $C$ being infra-dual has the form $-\min_{x,y\in\operatorname*{ex}B}\left\langle x_{m}
,y_{m}\right\rangle _{0}=-\min_{x,y\in\partial B}\left\langle x_{m}
,y_{m}\right\rangle _{0}\leq\mu$, as desired.

\item Let $\mu\leq-\max_{x\in\partial B}\min_{y\in\partial B}\left\langle
x_{m},y_{m}\right\rangle _{0}$. Suppose, contrary to our claim, that there
exists $x\in V$ such that $\langle x,y\rangle_{m,\mu}\geq0$ for all $y\in C$,
yet $x\notin C$. Then, putting $y=m$, we get $e(x)\geq0$. But also $\langle x+\lambda m,y\rangle_{m,\mu}=\langle x,y\rangle_{m,\mu}+\lambda\mu e(y)\geq0$ for any $\lambda>0$ and $y\in C$. Since $C$ is closed, there exists $\lambda>0$ such that $x+\lambda m\notin C$. Thus, we can actually
assume that $e(x)=1$ and $\left\langle
x-m,y-m\right\rangle _{0}\geq-\mu$ for all $y\in B$. Consequently, $\min_{y\in
B}\left\langle x-m,y-m\right\rangle _{0}\geq-\mu$. Take $v\in\arg\min_{y\in
B}\langle x-m,y-m\rangle_{0}$ and $t\in(0,1)$ such that $v^{\prime
}:=tx+(1-t)m\in\partial B$. Then $v\in\partial B$ and $v\in\arg\min_{y\in B}\langle
v^{\prime}-m,y-m\rangle_{0}$ since $x,v'$ and $m$ are collinear. Therefore
\begin{align*}
\langle x,v\rangle_{m,\mu}  & =\langle x-m,v-m\rangle_{0}+\mu=\frac{1}
{t}\langle v^{\prime}-m,v-m\rangle_{0}+\mu\\
& <\langle v^{\prime}-m,v-m\rangle_{0}+\mu=\min_{y\in\partial B}\langle
v^{\prime}-m,y-m\rangle_{0}+\mu\\
& \leq\max_{z\in\partial B}\min_{y\in\partial B}\langle z-m,y-m\rangle_{0}
+\mu\leq-\mu+\mu=0\text{,}
\end{align*}
a contradiction.

Now, let us assume that, contrary to our claim, $-\mu<\max_{x\in\partial B}
\min_{y\in\partial B}\left\langle x_{m},y_{m}\right\rangle _{0}$. Then, there
exists $x\in\partial B$ such that $0<\min_{y\in\partial B}\langle x,y\rangle_{m,\mu}=:\gamma$. Let $x'=tx+(1-t)m$ for some $t>1$. Since $m\in\textnormal{int}_{V_1}B$, $x'\notin B$. But $\langle x',y\rangle_{m,\mu}=t\langle x,y\rangle_{m,\mu}+(1-t)\mu\geq t\gamma+(1-t)\mu\geq 0$ for $t\in(1,\mu/(\mu-\gamma)$ and $y\in B$, a contradiction with the supra-duality assumption.
\item This follows from (i.) and (ii.).\qedhere
\end{enumerate}
\end{proof}

\section{Various interpretations of the space constant}\label{Appendix constant}
For some GGPTs the space constant $\chi / \mu$ plays the role of their \textsl{orthogonal dimension} minus $1$. Let us define an \textsl{orthogonal system} to be a set $\Omega$ of mutually orthogonal (in the sense of $\langle \cdot,\cdot\rangle _{m,\mu}$) elements from $M(B)$ such that $m \in \operatorname{aff} (\Omega)$. Note that this definition coincides with the one of orthogonal frame defined for self-dual GGPTs (see page 8). However, the  orthogonal frames are canonically defined by their operational interpretation as sets of perfectly distinguishable states. Since this interpretation gets lost without the self-duality assumption, we introduce the notion of orthogonal systems  in general setup.  Clearly, the cardinality of every orthogonal system is less or equal to $\dim V$. Then the following simple fact implies that the orthogonal systems (if exist) are equinumerous with $m$ at their center.

\begin{Prop}[orthogonal dimension]
\label{ort}
If $\Omega:= \{\omega_{i}\}_{i=1}^{M} \subset M(B)$ is 
\begin{enumerate}[i.]
\item
an orthogonal set, then $M - 1 \leq \chi/\mu$;
\item
an orthogonal system, then $ M-1 = \chi/\mu$ and $m = \frac{1}{M}\sum_{i=1}^{M}\omega_{i}$.
\end{enumerate}
\end{Prop}

\begin{proof}
To prove (i), set $m':= \frac{1}{M}\sum_{i=1}^{M}\omega_{i} \in B$. Multiplying this equality (in the sense of $\langle \cdot,\cdot\rangle _{m,\mu}$) by
$\omega_{k}$ ($k=1,\dots,M$), summing over $k$, and applying Proposition \ref{ProInn} we get $\mu \leq \langle m',m' \rangle=(\chi+\mu)/M$, as desired. Suppose now that $m=\sum_{i=1}^M\lambda_i\omega_i$, $\sum_{i=1}^M\lambda_i=1$. Multiplying by $\omega_k$ again, we get $\mu = \lambda_k (\chi + \mu)=(\chi+\mu)/M$, and (ii) follows.
\end{proof}

Besides, the coefficient $\chi/\mu$ can be interpreted in the language of thermodynamics as the \textsl{exponent of maximal entropy of the system} minus $1$. Introduce two entropies for GGPTs in the spirit of the \cite{Kruetal17} approach. Let $x\in B$. Define the \textsl{decomposition entropy} of $x$ as
$S_{2}(x):=\inf(-\ln\sum_{i=1}^{M}\lambda_{i}^{2})$, where the infimum is
taken over all decompositions $x=\sum_{i=1}^{M}\lambda_{i}\omega_{i}$,
$\omega_{i}\in\operatorname*{ex}B$, $\sum_{i=1}^{M}\lambda_{i}=1$, $\lambda_i \geq 0$, and
the \textsl{spectral entropy} of $x$ as $\widehat{S}_{2}(x):=\inf(-\ln\sum
_{i=1}^{M}\lambda_{i}^{2})$, where the infimum is taken over all such orthogonal decompositions, i.e. such that the states $\{\omega_i \}_{i=1}^M$ are orthogonal. Note that the
decompositions in the former case always exist according to the
Carath\'{e}odory theorem, but in the latter one this is not necessarily true. We call GGPTs \textsl{spectral} if every state can be represented as a convex combination of elements of some orthogonal system. 
Clearly, $S_{2}(x)\leq\widehat{S}_{2}(x)$.

\begin{Prop}[decomposition entropy]
\label{decent}
Let a GGPT be infra-dual and equinorm. Then for $x \in B$ we have $-\ln(\left\|
x\right\|  _{m,\mu}^{2}/(\chi+\mu))\leq S_{2}(x)$. In particular, $\ln(1+\chi/\mu)\leq S_{2}(m)$.
\end{Prop}

\begin{proof}
We have $\left\|  x\right\|  _{m,\mu}^{2}=\langle \sum_{i=1}
^{M}\lambda_{i}\omega_{i},\sum_{j=1}^{M}\lambda_{j}\omega_{j}\rangle
_{m,\mu}=\sum_{i,j=1}^{M}\lambda_{i}\lambda_{j}\left\langle \omega_{i},\omega
_{j}\right\rangle _{m,\mu}\geq\sum_{i=1}^{M}\lambda_{i}^{2}\left\|  \omega
_{i}\right\|  _{m,\mu}^{2}=(\chi+\mu)\sum_{i=1}^{M}\lambda_{i}^{2}$, where $x=\sum_{i=1}^{M}\lambda_{i}\omega_{i}$ is any decomposition of $x$ into
pure states.  In consequence,
$\ln((\chi+\mu)/\left\|  x\right\|  _{m,\mu}^{2})\leq-\ln\sum_{i=1}^{M}\lambda_{i}^{2}$, as desired.
\end{proof}

\begin{Prop}[spectral entropy]
\label{speent}
Let a GGPT be spectral. Then

\begin{enumerate}[i.]
\item $\widehat{S}_{2}(x)=-\ln(\left\|  x\right\|  _{m,\mu}^{2}/(\chi+\mu))=-\ln
\sum_{i=1}^{M}\lambda_{i}^{2}$ for any convex decomposition of $x$ into orthogonal system, $x=\sum_{i=1}^{M}\lambda_{i}\omega_{i}$.

\item $\widehat{S}_{2}(x)\leq\widehat{S}_{2}(m)=\ln(1+\chi/\mu)$ and $\widehat{S}
_{2}(x)=\widehat{S}_{2}(m)$ if and only if $x=m$.
\end{enumerate}
\end{Prop}

\begin{proof}
We have $\left\|  x\right\|  _{m,\mu}^{2}=\langle \sum_{i=1}^{M}\lambda
_{i}\omega_{i},\sum_{j=1}^{M}\lambda_{j}\omega_{j}\rangle _{m,\mu}
=\sum_{i,j=1}^{M}\lambda_{i}\lambda_{j}\langle \omega_{i},\omega
_{j}\rangle _{m,\mu}=(\chi+\mu)\sum_{i=1}^{M}\lambda_{i}^{2}$. This implies (i.).
Now, (ii.) follows from (i.) and Proposition \ref{ProInn}.\ref{norms}.
\end{proof}

From last three propositions it follows that if a GGPT is infra-dual and spectral (therefore necessarily equinorm), then both entropies coincide for every state and reach maximum equal natural logarithm of the orthogonal dimension of the GGPT at $m$, see also \cite[Theorem 11]{Kruetal17}. This justifies calling $m$ the \textsl{maximal
entropy state} or the \textsl{maximally mixed state}.

\begin{R}
The notion of the decomposition entropy is independent on the geometry of the state space and actually can be defined for any GPT, while the spectral entropy is strictly related to the geometry of the whole space (it is both $m$ and $\mu$-dependent). That is why for the spectral entropy we get the direct formula in Proposition \ref{speent}, whereas the geometry provides us just some bounds on the decomposition entropy in Proposition \ref{decent}. Note that the assumption that GGPT is equinorm fixes $m$ (and therefore $\chi$), but infra-duality just provides a lower bound on $\mu$, namely $\mu\geq\mu_i$. Therefore it is possible to improve the bounds in Proposition \ref{decent} by replacing $\mu$ with $\mu_i$. 
\end{R}

Finally, let us present the proof of Proposition \ref{extreme}:
\begin{proof}[Proof of Proposition \ref{extreme}]
The left-hand inequality and the assertion (i.) follows from Remark \ref{innout}. To prove the right-hand inequality and the assertion (ii.) take some
orthogonal frame with $M$ elements. From Proposition \ref{ort} we get
$\chi/\mu=M-1 \leq \dim V - 1 = \dim{V_0}$. On the other hand, if $\chi/\mu = \dim V_0$, then $M=\dim V$. Moreover, $B$ is a self-dual set containing a self-dual regular simplex. Hence (ii.) follows.
\end{proof}

\end{document}